\documentclass[11pt, a4paper]{article}
\usepackage[textwidth = 7in, textheight=25cm,nohead]{geometry}
\usepackage{ulem}
\usepackage{graphicx, amssymb, amsmath, amsthm, amstext, booktabs, float, multirow, subfigure, rotating,natbib}
\usepackage{setspace}
\numberwithin{equation}{section} 
\numberwithin{figure}{section} 
\numberwithin{table}{section} 
 
\newtheorem{lemma}{Lemma}[section] 

\newtheorem{theorem}{Theorem}[section] 
\newtheorem{assumption}{Assumption}[section] 
\newtheorem{Proposition}{Proposition}[section]

\bibpunct{(}{)}{,}{a}{,}{,}

\def\ppn{\vskip 6pt \noindent }
\def\R{{\mathbb{R}}}

\def\P{{\mathbb{P}}}
\def\E{{\mathbb{E}}}
\newcommand{{\Xs}}{{\cal X}}
\newcommand{{\Ls}}{{\cal L}}
\newcommand{{\Ss}}{{\cal S}}
\newcommand{{\Gs}}{{\cal G}}
\newcommand{{\Hs}}{{\cal H}}
\newcommand{{\Ns}}{{\cal N}}
\newcommand{{\Is}}{{\cal I}}
\newcommand{{\Bs}}{{\cal B}}
\newcommand{{\Cs}}{{\cal C}}
\newcommand{{\Rs}}{{\cal R}}
\newcommand{{\Us}}{{\cal U}}
\newcommand{{\pp}}{{\mathbf p}}
\newcommand{{\KK}}{{\mathbf K}}
\newcommand{{\WW}}{{\mathbf W}}
\newcommand{{\HH}}{{\mathbf H}}
\newcommand{{\II}}{{\mathbf I}}
\newcommand{{\yy}}{{\mathbf y}}
\newcommand{{\ab}}{{\mathbf a}}
\newcommand{\C}{\mathbb{C}}
\newcommand{\G}{\mathbb{G}}
\newcommand{\B}{\mathbb{B}}

\newcommand{{\toL}}{{\overset{\mathcal{L}}{\longrightarrow\ }}}

\newcommand{{\dou}}{$\leadsto$\ }

\newcommand{{\phphu}}{{\phi(\Phi^{-1}(u))}}
\newcommand{{\phphv}}{{\phi(\Phi^{-1}(v))}}
\newcommand{{\phphUi}}{{\phi(\Phi^{-1}(\hat{U}_i))}}
\newcommand{{\phphVi}}{{\phi(\Phi^{-1}(\hat{V}_i))}}

\DeclareMathOperator{\var}{\mathbb{V}ar}

\DeclareMathOperator{\diag}{diag}

\newcommand{\indic}[1]{
\hbox{${\it 1}\hskip -4.5pt I_{\{ #1 \}}$}
}

\begin{document}
\title{Probit transformation for nonparametric kernel estimation of the copula density
}
\author{\sc{Gery Geenens}\thanks{Corresponding author: ggeenens@unsw.edu.au, School of Mathematics and Statistics, University of New South Wales, Sydney, NSW 2052 (Australia), tel +61 2 938 57032, fax +61 2 9385 7123 }\\School of Mathematics and Statistics,\\ University of New South Wales, Sydney, Australia,  \and \sc{Arthur Charpentier} \\ D\'epartement de math\'ematiques, \\ Universit\'e du Qu\'ebec \`a Montr\'eal, Canada,  \and and \sc{Davy Paindaveine} \\ECARES,\\ Universit\'e Libre de Bruxelles, Belgium}
\date{\today}
\maketitle
\thispagestyle{empty} 

\begin{abstract}

\noindent Copula modelling has become ubiquitous in modern statistics. Here, the problem of nonparametrically estimating a copula density is addressed. Arguably the most popular nonparametric density estimator, the kernel estimator is not suitable for the unit-square-supported copula densities, mainly because it is heavily affected by boundary bias issues. In addition, most common copulas admit unbounded densities, and kernel methods are not consistent in that case. In this paper, a kernel-type copula density estimator is proposed. It is based on the idea of transforming the uniform marginals of the copula density into normal distributions via the probit function, estimating the density in the transformed domain, which can be accomplished without boundary problems, and obtaining an estimate of the copula density through back-transformation. Although natural, a raw application of this procedure was, however, seen not to perform very well in the earlier literature. Here, it is shown that, if combined with 
local likelihood density estimation methods, the idea yields very good and easy to implement estimators, fixing boundary issues in a natural way and able to cope with unbounded copula densities. The asymptotic properties of the suggested estimators are derived, and a practical way of selecting the crucially important smoothing parameters is devised. Finally, extensive simulation studies and a real data analysis evidence their excellent performance compared to their main competitors. 

\noindent \textbf{Keywords:} copula density; transformation kernel density estimator; boundary bias; unbounded density; local likelihood density estimation.
\end{abstract}

\newpage

\section{Introduction} \label{sec:intro}

For the last two decades copula modelling has emerged as a major research area of statistics. By definition, a bivariate copula function $C$ is the joint cumulative distribution function (often abbreviated to `cdf' below) of a bivariate random vector whose marginals are Uniform over $[0,1]$, i.e., 
\[C: \Is \doteq [0,1]^2 \to [0,1]: (u,v) \to C(u,v) = \P(U \leq u, V \leq v), \]
where $U \sim \Us_{[0,1]}$, $V \sim \Us_{[0,1]}$. Copulas arise naturally in statistics and probability as a mere consequence of two well-known facts. First, the {\it probability-integral transform} result, establishing that for any continuous variable $X$ with distribution $F_X$, $F_X(X) \sim  \Us_{[0,1]}$, and second, {\it Sklar's theorem} \citep{Sklar59}, stating that for any continuous bivariate distribution whose cdf is $F_{XY}$, there exists a unique function $C$ such that
\begin{equation} F_{XY}(x,y) = C(F_X(x),F_Y(y)) \qquad \forall (x,y) \in \R^2,  \label{eqn:copdef} \end{equation}
where $F_X$ and $F_Y$ are the marginals of $F_{XY}$. According to the previous definition, this function $C$ is, indeed, a copula, called the copula of $F_{XY}$. From (\ref{eqn:copdef}) it is clear that $C$ describes how the two marginal distributions $F_X$ and $F_Y$ `interact' to produce the joint $F_{XY}$. It, therefore, disjoints the marginal behaviours of $X$ and $Y$ from their dependence structure, hence the attractiveness of the copula approach. See \cite{Joe97} and \cite{Nelsen06} for book length treatment of the foregoing ideas. Other, more compact reviews include \cite{Genest07,Hardle09} and \cite{Embrechts09}. Today, copulas are used extensively in statistical modelling in all areas, from quantitative finance and insurance to medicine and climatology. Therefore, empirically estimating a copula function from an observed bivariate sample $\{(X_{i},Y_{i});i=1,\ldots,n\}$ drawn from 
$F_{XY}$ has become an important problem of modern statistical modelling.

\ppn Of course, estimating $C$ essentially amounts to fitting a bivariate distribution, for what many parametric families have been suggested and studied: Gaussian, Student-$t$, Clayton, Frank or Gumbel copulas among others (see again \cite{Joe97} or \cite{Nelsen06} for details). These parametric models have formed the main body of the literature in the field so far. However, they suffer from the usual lack of flexibility of parametric approaches and the induced risk of misspecification. For instance, it has been argued that the main reason behind the 2009 global financial crisis was a reckless usage of the Gaussian copula \citep{Salmon09}. There is, therefore, a tremendous need for flexible nonparametric copula models, making no rigid assumptions on the underlying distributions. An early step in that direction was the empirical copula devised by \cite{Deheuvels79}. The related empirical copula process was studied further in \cite{Fermanian04,Tsukuhara05,Segers12} and \cite{Bucher13}, and turns out to be the 
cornerstone of a variety of nonparametric copula-based procedures, see e.g.\ \cite{Genest04,Genest09a,Gudendorf12} or \cite{Li13}, to cite only a few. Moreover, \cite{Fermanian03}, \cite{Chen07}, \cite{Omelka09} and \cite{Gijbels10} studied kernel methods to obtain flexible smooth estimates of the bivariate cdf $C$. 

\ppn It is usually the case, though, that a distribution is more readily interpretable in terms of its probability density function than directly in terms of its cdf, and a copula is, in many aspects, no different. Assume that the bivariate cdf $C$ is absolutely continuous. Then, its associated density is
\[c(u,v) = \frac{\partial^2 C}{\partial u \partial v }(u,v) \]
for $(u,v) \in \Is$, a function naturally enough called the {\it copula density}. This paper precisely addresses the problem of estimating this copula density $c$ in a nonparametric way, for what kernel methods again appear natural. This approach was pioneered in \cite{Behnen85} and \cite{Gijbels90}, and arguably remains very attractive compared to its competitors, such as splines smoothing \citep{Shen08,Kauermann13}, wavelets \citep{Hall06,Genest09b,Autin10}, Bernstein polynomials \citep{Bouezmarni10,Bouezmarni11,Janssen13} or others \citep{Qu12}, for its simplicity in construction, implementation and interpretation.

\ppn At least three factors make kernel estimation of $c(u,v)$ not standard, though, and have delayed the development of reliable kernel copula density estimators. A major concern is that kernel estimators suffer from boundary bias problems. Given the bivariate sample $\{(U_i=F_X(X_i),V_i=F_Y(Y_i));i=1,\ldots,n\}$, the standard kernel estimator for $c$, say $\hat{c}^*$, at $(u,v) \in \Is$ would be \citep[Chapter 4]{Wand95}
\begin{equation} \hat{c}^*(u,v) = \frac{1}{n|\HH_{UV}|^{1/2}} \sum_{i=1}^n \KK\left(\HH_{UV}^{-1/2}\binom{u-U_i}{v-V_i} \right),   \label{eqn:kerncopbasstar}\end{equation}
where $\KK: \R^2 \to \R$ is a bivariate kernel function and 
$\HH_{UV}$ is a symmetric positive-definite bandwidth matrix. It is, however, well known that an estimator such as (\ref{eqn:kerncopbasstar}) is in general not consistent on the boundaries of $\Is$: it does not `feel' the support boundaries of the underlying density and places through $\KK$ positive mass outside that support. In fact, standard kernel density arguments show that $\E(\hat{c}^*(u,v)) = \frac{1}{4} c(u,v) + O(h)$ at corners ($(u,v) \in \{(0,0),(1,0),(0,1),(1,1)\}$) and  $\E(\hat{c}^*(u,v))=\frac{1}{2} c(u,v) + O(h)$ on the borders ($(u,v) \in \{(0,\nu),(1,\nu),(\mu,0),(\mu,1): \mu,\nu \in (0,1) \}$). Although some papers ignored these boundary issues \citep{Fermanian03,Fermanian05,Scaillet07,Faugeras09}, it is clear that accurate estimation of $c$ calls for some boundary correction. Such corrections have indeed been proposed, inspired by ideas developed for univariate density estimation, e.g.\ mirror reflection \citep{Gijbels90} or the usage of boundary kernels \citep{Chen07}, but with mixed results. 

\ppn Secondly, kernel estimators are not consistent for unbounded densities. Yet, unlike most common probability densities, many copula densities of interest are unbounded. For instance, even in the apparently easy case of a bivariate Normal vector with moderate correlation, the copula density is unbounded in two of the corners of $\Is$. It is, therefore, particularly important to use estimators able to cope with such unboundedness. Finally, estimating $c$ cannot be made from a genuine random sample from its cdf $C$, as $C$ is the distribution of $(U,V)=(F_X(X),F_Y(Y))$ and $F_X$ and $F_Y$ are typically unknown. Hence, the observations $(U_i,V_i)$ are unavailable, and estimator (\ref{eqn:kerncopbasstar}) is, in fact, infeasible. In the copula literature, it is customary to use the `pseudo-observations' 
\begin{equation} \hat{U}_i = \frac{n}{n+1} \hat{F}_{Xn}(X_{i}) \qquad \text{ and } \qquad  \hat{V}_i = \frac{n}{n+1} \hat{F}_{Yn}(Y_{i}) \label{eqn:pseudobs} \end{equation}
where $\hat{F}_{Xn}(x) = \frac{1}{n}\sum_{j=1}^n \indic{X_{j}\leq x}$ is the empirical cdf of $X$, and similarly for $\hat{F}_{Yn}$. The rescaling by $n/(n+1)$ in (\ref{eqn:pseudobs}), aiming at keeping the pseudo-observations in the interior of $[0,1]$, is also common practice. Then, the pseudo-sample $\{(\hat{U}_i,\hat{V}_i);i=1,\ldots,n\}$ is treated mostly as a sample from $C$ and used instead of the `true' sample $\{(U_i,V_i);i=1,\ldots,n\}$, although this may affect the statistical properties of the ensuing estimators \citep{Charpentier07,Genest10}. 

\ppn The aim of this paper is to propose and study a new, kernel-type estimator of the copula density $c$. It is, in fact, the extension to the copula density case of the kernel-type estimator for univariate densities supported on the unit interval recently suggested in \cite{Geenens13}. That estimator takes the constrained nature of the support into account from the outset, i.e.\ without relying on {\it ad hoc} boundary corrections (reflection, boundary kernels, etc.). It proved superior to its main competitors in the simulation studies for a wide range of density shapes, including for unbounded densities. The idea seems, therefore, suitable for estimating copula densities as well. Specifically, \cite{Geenens13}'s estimator makes use of the transformation method, building on ideas first suggested in \citet[Chapter 9]{Devroye85} and \cite{Marron94}. In short, the initial $[0,1]$-supported variable of interest is transformed through the probit function into a variable whose support is unconstrained, the 
density of that transformed variable is estimated and an estimate of the initial density on $[0,1]$ is obtained by back-transformation. This method appears very natural and yields very good results, provided that the estimation step in the transformed domain is carried out with care, as \cite{Geenens13}'s results showed.

\ppn Exploring this idea in more details in the context of copula density estimation is the topic of Section \ref{sec:probtransfcop}. Several versions of the estimators will be suggested, and their asymptotic properties will be derived in Section \ref{sec:asympt}. Section \ref{sec:bandwidth} will address the crucial point of smoothing parameter selection in this framework. 
Simulation studies evidencing the very good practical behaviour of the probit-transformation estimators (Section \ref{sec:sim}), a real data analysis (Section \ref{sec:realdat}) and some final remarks (Section \ref{sec:ccl}) conclude the paper.

\section{Probit transformation kernel copula density estimation} \label{sec:probtransfcop}

\subsection{Probit transformation} \label{subsec:probit}

As recalled above, direct kernel estimation of the density $c$ of $(U,V)$ is made difficult mainly by the constrained nature of its support $\Is=[0,1]^2$. Now, define 
\[S = \Phi^{-1}(U) \qquad \text{ and } \qquad T = \Phi^{-1}(V),  \]
where $\Phi$ is the standard normal cdf and $\Phi^{-1}$ is its quantile function (i.e.\ the probit transformation). Given that both $U$ and $V$ are $\Us_{[0,1]}$, $S$ and $T$ both follow standard normal distributions, which, nevertheless, does not imply that the vector $(S,T)$ is bivariate normal. That will only be the case if the copula of the joint cdf of $(S,T)$, say $F_{ST}$, is the Gaussian copula, that is, if the copula $C$ of $F_{XY}$ itself is the Gaussian copula, as copulas are invariant to increasing transformations of their margins \citep[Theorem 2.4.3]{Nelsen06}. The idea is that, if $c(u,v)>0$ Lebesgue-a.e.\ over $\Is$ (which will be assumed throughout the paper), $(S,T)$ has unconstrained support $\R^2$ and estimating its density $f_{ST}$ cannot suffer from boundary issues. In addition, due to its normal margins, one can expect $f_{ST}$ to be well-behaved, and its estimation easy. In particular, under mild assumptions, $f_{ST}$ and its partial derivatives up to the second order will be seen to 
be uniformly bounded on $\R^2$, even in the case of unbounded copula density $c$ (Lemma \ref{lem:unicont} in the Appendix).

\ppn As the copula of $F_{ST}$ is $C$, $S \sim \Ns(0,1)$ and $T \sim \Ns(0,1)$, one can write Sklar's theorem (\ref{eqn:copdef}) for $(S,T)$:
\[F_{ST}(s,t)=\P(S \leq s, T \leq t) = C\left(\Phi(s),\Phi(t) \right), \qquad \forall (s,t) \in \R^2. \]
Upon differentiation with respect to $s$ and $t$, the joint density $f_{ST}$ of $(S,T)$ is found to be
\begin{equation} f_{ST}(s,t) = c(\Phi(s),\Phi(t))\phi(s)\phi(t) , \label{eqn:fST} \end{equation}
where $\phi$ is the standard normal density.  Inverting this expression, one obtains
\begin{equation} c(u,v) = \frac{f_{ST}(\Phi^{-1}(u),\Phi^{-1}(v))}{\phi(\Phi^{-1}(u))\phi(\Phi^{-1}(v))} \label{eqn:truecopratio} \end{equation}
for any $(u,v) \in (0,1)^2$. So, any estimator $\hat{f}_{ST}$ of $f_{ST}$ on $\R^2$ automatically produces an estimator of the copula density on the interior of $\Is$, viz.
\begin{equation} \hat{c}^{(\tau)}(u,v) = \frac{\hat{f}_{ST}(\Phi^{-1}(u),\Phi^{-1}(v))}{\phi(\Phi^{-1}(u))\phi(\Phi^{-1}(v))} \label{eqn:copratio} \end{equation}
where the superscript $(\tau)$ refers to the idea of transformation. When necessary, $\hat{c}^{(\tau)}$ can also be defined at the boundaries of $\Is$ by continuity. 
This estimator enjoys interesting properties. Clearly, $\hat{c}^{(\tau)}$ cannot allocate any positive probability weight outside $\Is$, since $(\Phi^{-1}(u),\Phi^{-1}(v))$ is not defined for $(u,v) \notin \Is$. Also, if $\hat{f}_{ST}$ is a {\it bona fide} density function, in the sense that $\hat{f}_{ST}(s,t) \geq 0$ for all $(s,t)$ and $\iint_{\R^2} \hat{f}_{ST}(s,t)\,ds\,dt = 1$, then automatically $\hat{c}^{(\tau)}(u,v) \geq 0$ for all $(u,v) \in \Is$ and $\iint_\Is \hat{c}^{(\tau)}(u,v)\,du\,dv = 1$. This is easily seen through the changes of variable $u=\Phi(s)$ and $v=\Phi(t)$. Finally, if $\hat{f}_{ST}$ is a uniformly (weak or strong) consistent estimator for $f_{ST}$, i.e.\ $\sup_{(s,t) \in \R^2} |\hat{f}_{ST}(s,t)-f_{ST}(s,t)| \to 0$ in probability or almost surely as $n \to \infty$, the estimator $\hat{c}^{(\tau)}$ inherits that uniform consistency on any compact proper subset of $\Is$. 

\subsection{The naive estimator}

A first natural idea would be to use the standard kernel density estimator as $\hat{f}_{ST}$ in (\ref{eqn:copratio}). Specifically, one would like to use an estimator like
\begin{equation} \hat{f}^*_{ST}(s,t) = \frac{1}{n|\HH_{ST}|^{1/2}} \sum_{i=1}^n \KK\left(\HH_{ST}^{-1/2}\binom{s-S_i}{t-T_i} \right) \label{eqn:fSTstar} \end{equation}
where $\KK$ is a bivariate kernel function and $\HH_{ST}$ is some symmetric positive-definite bandwidth matrix, and 
$\{(S_i = \Phi^{-1}(U_i),T_i=\Phi^{-1}(V_i)); i=1,\ldots,n\}$ is the sample in the transformed domain. However, as the $(U_i,V_i)$'s are unavailable in this context, so are the $(S_i,T_i)$'s. Instead, one has to use 
\begin{equation} \{(\hat{S}_i=\Phi^{-1}(\hat{U}_i),\hat{T}_i=\Phi^{-1}(\hat{V}_i));i=1,\ldots,n\}, \label{eqn:pseudosamp} \end{equation}
the pseudo-transformed sample. The feasible version of (\ref{eqn:fSTstar}) is, therefore,
\begin{equation} \hat{f}_{ST}(s,t) = \frac{1}{n|\HH_{ST}|^{1/2}} \sum_{i=1}^n \KK\left(\HH_{ST}^{-1/2}\binom{s-\hat{S}_i}{t-\hat{T}_i} \right). \label{eqn:fSThat} \end{equation}
Through (\ref{eqn:copratio}), this directly leads to the following probit transformation kernel copula density estimator: 
\begin{equation} \hat{c}^{(\tau)}(u,v) = \frac{1}{n|\HH_{ST}|^{1/2}\phphu\phphv} \sum_{i=1}^n \KK\left(\HH_{ST}^{-1/2}\binom{\Phi^{-1}(u)-\Phi^{-1}(\hat{U}_i)}{\Phi^{-1}(v)-\Phi^{-1}(\hat{V}_i)} \right). \label{eqn:naivefeas}\end{equation}
This is essentially the estimator suggested in \cite{Charpentier07}, also used as-is in \cite{LopezPaz13}, although it was not studied in any details in those two papers. \cite{Omelka09} derived the theoretical properties of an estimator for the copula $C$ (not its density) based on the same transformation. 

\ppn This idea was, however, called `naive' in \cite{Geenens13} in the univariate case. In fact, the method does not provide good results close to the boundaries, even though it was designed to fix boundary issues. Indeed, the estimator (\ref{eqn:naivefeas}) will be seen not to perform well in the next sections. \cite{Geenens13} explained the reasons for that  failure, and suggested some remedies. In particular, estimating the density in the transformed domain via local likelihood methods \citep{Loader96,Hjort96} offers a promising alternative while keeping the simplicity and the intuitive appeal of the probit-transformation estimator. This is investigated for estimating a copula density in the next subsection.

\subsection{Improved probit-transformation copula density estimators}

\cite{Loader96} and \cite{Hjort96} proposed two different, although similar in many aspects, formulations of the local likelihood density estimator. \cite{Loader96} locally approximates the logarithm of the unknown density by a polynomial, whereas \cite{Hjort96} consider local parametric density modelling. This paper will only make use of \cite{Loader96}'s idea, mainly because the asymptotic theory is more transparent. In any case, both formulations share most of their advantages and drawbacks, and typically yield very similar estimates. 

\ppn In this setting of estimating $f_{ST}$ from the pseudo-sample $\{(\hat{S}_i,\hat{T}_i);i=1,\ldots,n\}$, \cite{Loader96}'s local likelihood estimator is defined as follows. Around $(s,t) \in \R^2$, $\log f_{ST}$ is assumed to be well approximated by a polynomial of some order $p$. Classically, only local log-linear ($p=1$) and local log-quadratic ($p=2$) estimators are considered. Specifically, in the first case ($p=1$), it is assumed that, for $(\check{s},\check{t})$ `close' to $(s,t)$, 
\begin{equation} \log f_{ST}(\check{s},\check{t}) \simeq a_{1,0}(s,t) + a_{1,1}(s,t)(\check{s}-s) + a_{1,2}(s,t)(\check{t}-t) \doteq P_{\ab_1}(\check{s}-s,\check{t}-t)\label{eqn:locloglin}  \end{equation}
and in the second case ($p=2$)
\begin{align*} \log f_{ST}(\check{s},\check{t})  \simeq & \ a_{2,0}(s,t) + a_{2,1}(s,t)(\check{s}-s) + a_{2,2}(s,t)(\check{t}-t) \\ & + a_{2,3}(s,t)(\check{s}-s)^2 + a_{2,4}(s,t)(\check{t}-t)^2 + a_{2,5}(s,t)(\check{s}-s)(\check{t}-t) \\ &\ \doteq P_{\ab_2}(\check{s}-s,\check{t}-t) .\end{align*}
The vectors $\ab_1(s,t)=(a_{1,0}(s,t),a_{1,1}(s,t),a_{1,2}(s,t))$ and $\ab_2(s,t) \doteq (a_{2,0}(s,t),\ldots,a_{2,5}(s,t))$ are then estimated by solving a weighted maximum likelihood problem. For either $p=1,2$,
\begin{multline} \tilde{\ab}_p(s,t) = \arg \max_{\ab_p} \left\{ \sum_{i=1}^n \KK\left(\HH_{ST}^{-1/2}\binom{s-\hat{S}_i}{t-\hat{T}_i} \right) P_{\ab_p}(\hat{S}_i-s,\hat{T}_i-t)\right. \\ \left.  - n\iint_{\R^2} \KK\left(\HH_{ST}^{-1/2}\binom{s-\check{s}}{t-\check{t}} \right) \exp\left(P_{\ab_p}(\check{s}-s,\check{t}-t)\right)\,d\check{s}\,d\check{t} \right\}, \label{eqn:loclogpol} \end{multline}
where, as previously, $\KK$ is a bivariate kernel function and $\HH_{ST}$ is a symmetric positive-definite bandwidth matrix. The estimate of $f_{ST}$ at $(s,t)$ is then, naturally, $\tilde{f}^{(1)}_{ST}(s,t) = \exp(\tilde{a}_{1,0}(s,t))$ for local log-linear, and $\tilde{f}^{(2)}_{ST}(s,t) = \exp(\tilde{a}_{2,0}(s,t))$ for local log-quadratic modelling. 
`Improved' probit-transformation kernel copula density estimators for $c(u,v)$ are finally obtained through (\ref{eqn:copratio}) as
\begin{equation} \tilde{c}^{(\tau,p)}(u,v) = \frac{\tilde{f}^{(p)}_{ST}(\Phi^{-1}(u),\Phi^{-1}(v))}{\phi(\Phi^{-1}(u))\phi(\Phi^{-1}(v))} \label{eqn:impctilde} \end{equation}
for $p=1$ and $p=2$. The motivation and the advantages of estimating $f_{ST}$ by local likelihood methods instead of raw kernel density estimation are related to the detailed discussion in \cite{Geenens13}, and are therefore omitted here. The asymptotic properties of these estimators (`naive' and `improved' probit-transformation kernel copula density estimators) are derived in the next section.

\section{Asymptotic properties} \label{sec:asympt}

For simplicity, it will be assumed that $\KK$ is a product Gaussian kernel, i.e.\ $\KK(z_1,z_2) = \phi(z_1)\phi(z_2)$, and $\HH_{ST} = h^2 \II$ for some $h >0$. Note that, in practice, there are reasons to keep an unconstrained, non-diagonal bandwidth matrix $\HH_{ST}$. In particular, the copula density is typically stretched along one of the diagonals of the unit square when some dependence is present in $(X,Y)$, which provides a density $f_{ST}$ likewise stretched along one of the 45 degrees lines in $\R^2$. Hence, using a bandwidth matrix directing smoothing in that particular direction is sensible \citep{Duong05}, as discussed further in Section \ref{sec:bandwidth}. That said, theoretical results for that general case would be less tractable than, while qualitatively equivalent to, the simpler case presented below. Note that for that particular kernel $\KK$, $\iint \KK^2(z_1,z_2)\,dz_1\,dz_2 = (4\pi)^{-1}$ and $\iint z^2_k \KK(z_1,z_2)\,dz_1\,dz_2 = 1$, $k=1,2$. These quantities frequently arise in the properties of 
kernel estimators, and direct use of these particular numerical values will be made in the results below.

\subsection{The naive estimator and an amended version} \label{subsec:naive}

Consider the naive estimator (\ref{eqn:naivefeas}) which, with the above specifications of $\KK$ and $\HH_{ST}$, reduces to 
\begin{equation} \hat{c}^{(\tau)}(u,v) = \frac{1}{nh^2\phphu\phphv} \sum_{i=1}^n \phi\left(\frac{\Phi^{-1}(u)-\Phi^{-1}(\hat{U}_i)}{h}\right)\phi\left(\frac{\Phi^{-1}(v)-\Phi^{-1}(\hat{V}_i)}{h}\right). \label{eqn:naivefeas2} \end{equation}
Given (\ref{eqn:copratio}), it is clear that its statistical properties will entirely depend on those of (\ref{eqn:fSThat}), here
\begin{equation} \hat{f}_{ST}(s,t) = \frac{1}{nh^2} \sum_{i=1}^n \phi\left(\frac{s-\hat{S}_i}{h} \right) \phi\left(\frac{t-\hat{T}_i}{h}  \right).  \label{eqn:fSThat2} \end{equation}
If $f_{ST}$ admits continuous second-order partial derivatives, expressions for the bias and the variance of the ideal, infeasible estimator $\hat{f}^*_{ST}$ (\ref{eqn:fSTstar}), as well as its asymptotic normality, are well known \citep[Chapter 4]{Wand95}. Proposition \ref{thm:fShat} below ascertains that using the pseudo-observations (\ref{eqn:pseudosamp}) instead of genuine ones does not affect those properties. Note that (\ref{eqn:fSThat2}) can be written
\begin{equation} \hat{f}_{ST}(s,t)  = \frac{1}{h^2} \iint_{\R^2} \phi\left(\frac{s-\Phi^{-1}(u)}{h} \right) \phi\left(\frac{t-\Phi^{-1}(v)}{h}  \right) d\hat{C}_n(u,v), \label{eqn:fSThat3} \end{equation}
where $\hat{C}_n$ is the empirical copula
\begin{equation} \hat{C}_n(u,v) = \frac{1}{n} \sum_{i=1}^n \indic{\hat{U}_i \leq u,\hat{V}_i \leq v}. \label{eqn:empcop} \end{equation} 
Hence, although living in the transformed domain, the behaviour of $\hat{f}_{ST}(s,t)$ will be driven by the properties of $\hat{C}_n$ on $\Is$. The following assumptions will be made.
\begin{assumption} \label{ass:FXY} The sample $\{(X_{i},Y_{i});i=1,\ldots,n\}$ is an i.i.d.\ sample from the joint distribution $F_{XY}$, an absolutely continuous distribution with marginals $F_X$ and $F_Y$ strictly increasing on their support;
\end{assumption}
\begin{assumption} \label{ass:cop} The copula $C$ of $F_{XY}$ is such that $(\partial C/ \partial u)(u,v)$ and $(\partial^2 C/ \partial u^2)(u,v)$ exist and are continuous on $\{(u,v): u \in (0,1),v\in [0,1]\}$, and $(\partial C/ \partial v)(u,v)$ and $(\partial^2 C/ \partial v^2)(u,v)$ exist and are continuous on $\{(u,v): u \in [0,1],v\in (0,1)\}$. In addition, there are constants $K_1$ and $K_2$ such that
\begin{equation*} \left\{ \begin{array}{rcl} \displaystyle \left|\frac{\partial^2 C }{\partial u^2}(u,v) \right| &\leq & \displaystyle \frac{K_1}{u(1-u)} \qquad  \text{ for } (u,v) \in (0,1) \times [0,1]; \\
\displaystyle \left|\frac{\partial^2 C }{\partial v^2}(u,v) \right| &\leq & \displaystyle \frac{K_2}{v(1-v)} \qquad  \text{ for } (u,v) \in [0,1] \times (0,1); \end{array}\right. \end{equation*}
\end{assumption}
\begin{assumption} \label{ass:copdens}
The density $c$ of $C$ exists, is positive and admits continuous second-order partial derivatives on the interior of the unit square $\Is$. In addition, there is a constant $K_{00}$ such that
\begin{equation}  c(u,v)  \leq K_{00} \min\left(\frac{1}{u(1-u)},\frac{1}{v(1-v)} \right) \qquad \forall (u,v) \in (0,1)^2. \label{eqn:unbounded} \end{equation} 
\end{assumption}
Assumption \ref{ass:FXY} guarantees the existence and the uniqueness of the copula $C$ of $F_{XY}$. Assumptions \ref{ass:cop}-\ref{ass:copdens} mostly reduce to Conditions 2.1 and 4.1 in \cite{Segers12}, who claims that they hold for many copula families, such as Gaussian, Archimedean and most extreme-value copulas. Moreover, \cite{Omelka09} explicitly show that they are satisfied by the Clayton, Gumbel, Gaussian and Student copulas. Compared to \cite{Segers12}, Assumption \ref{ass:copdens} only requires further the existence and continuity of second-order partial derivatives of $c$, which is natural in kernel estimation. It is worth noting that $c$ is allowed to grow unboundedly in some of the corners of $\Is$, provided (\ref{eqn:unbounded}) remains valid.  The following result can now be stated. An important observation is that it holds true for $h \sim n^{-a}$, $a 
\in (0,1/4)$, which includes the optimal 
bandwidth order known to be $h \sim n^{-1/6}$ for bivariate density estimation.

\begin{Proposition} \label{thm:fShat} Assume that $\KK(z_1,z_2)=\phi(z_1)\phi(z_2)$ and $\HH_{ST} = h^2 \II$ with $h \sim n^{-a}$ for some $a \in (0,1/4)$. Under Assumptions \ref{ass:FXY}-\ref{ass:copdens}, the estimator (\ref{eqn:fSThat2}) at any $(s,t) \in \R^2$ is such that
\begin{equation} \sqrt{nh^2}\left(\hat{f}_{ST}(s,t) - f_{ST}(s,t)-h^2 b_{ST}(s,t)\right)  \toL \Ns\left(0 , \sigma_{ST}^2(s,t) \right), \label{eqn:normfSThat} \end{equation}
where $b_{ST}(s,t)  = \frac{1}{2}\left(\frac{\partial^2 f_{ST}}{\partial s^2}(s,t)+\frac{\partial^2 f_{ST}}{\partial t^2}(s,t) \right)$ and $\sigma_{ST}^2(s,t) = \frac{f_{ST}(s,t)}{4\pi}$.
\end{Proposition}
\begin{proof} See Appendix. \end{proof}
As recalled in Section \ref{sec:intro}, resorting to pseudo-observations is known to usually affect the statistical properties of the estimators of interest in copula modelling. In particular, an overriding result in the field is the weak convergence of the empirical copula process
\begin{equation} \C_n(u,v) \doteq \sqrt{n}(\hat{C}_n(u,v) - C(u,v)) \leadsto \G_C(u,v) \doteq \B_C(u,v) - \frac{\partial C}{\partial u}(u,v) \B_C(u,1) - \frac{\partial C}{\partial v}(u,v) \B_C(1,v), \label{eqn:limitempcop} \end{equation}
where $\B_C(u,v)$ is a bivariate pinned $C$-Brownian Sheet on $\Is$, i.e.\ the tight centred Gaussian process whose covariance function is $\E\left( \B_C(u,v)\B_C(u',v') \right) = C(u \wedge u',v \wedge v')-C(u,v)C(u',v')$ \citep{Fermanian04,Segers12}. In fact, $\B_C(u,v)$ would be the limiting process if the margins were known, i.e.\ if `genuine' $U_i$'s and $V_i$'s were used in (\ref{eqn:empcop}). The extra two terms in the right-hand side of (\ref{eqn:limitempcop}) are, therefore, often interpreted as `the price to pay' for using pseudo-observations -- although this effect is sometimes advantageous \citep{Genest10}. Yet, the proof of Proposition \ref{thm:fShat} reveals that the effect of those two terms asymptotically vanishes when one looks at the properties of the kernel density estimator (\ref{eqn:fSThat2}). As a result, the rate of convergence, as well as the expressions for asymptotic bias and variance, are the same as what one would obtain for the ideal estimator $\hat{f}^*_{ST}$ using genuine
i.i.d.\ observations from $C$. Intuitively, this is because a kernel density estimator converges slower than an empirical distribution function. Resorting to pseudo-observations may disturb the $\sqrt{n}$-convergence of the latter, but it goes unnoticed (asymptotically) compared to the nonparametric convergence rate $O((nh^2)^{-1/2})$ of the former.  

\ppn Now, differentiating (\ref{eqn:fST}) yields
\begin{align} \frac{\partial f_{ST}}{\partial s}(s,t) & = \frac{\partial c}{\partial u}(\Phi(s),\Phi(t))\phi^2(s) \phi(t) - sc(\Phi(s),\Phi(t))  \phi(s)\phi(t), \quad \text{ and} \label{eqn:firstderfST} \\
\frac{\partial^2 f_{ST}}{\partial s^2}(s,t) & =  \frac{\partial^2 c}{\partial u^2}(\Phi(s),\Phi(t))\phi^3(s) \phi(t) - 3s \frac{\partial  c}{\partial u}(\Phi(s),\Phi(t)) \phi^2(s) \phi(t)  +(s^2-1) c(\Phi(s),\Phi(t))  \phi(s)\phi(t) \label{eqn:secondderfST} \end{align}
(and similar for $\frac{\partial f_{ST}}{\partial t}$, $\frac{\partial^2 f_{ST}}{\partial t^2}$ and $\frac{\partial^2 f_{ST}}{\partial s \partial t}$). Hence, combining (\ref{eqn:fST}), (\ref{eqn:copratio}),  (\ref{eqn:normfSThat}) and (\ref{eqn:secondderfST}), one can state the following theorem for the `naive' probit transformation kernel copula density estimator (\ref{eqn:naivefeas2}).
\begin{theorem} \label{thm:chat} Under the assumptions of Proposition \ref{thm:fShat}, the `naive' probit transformation kernel copula density estimator (\ref{eqn:naivefeas2}) at any $(u,v) \in (0,1)^2$ is such that
\begin{equation*} \sqrt{nh^2}\left(\hat{c}^{(\tau)}(u,v) - c(u,v)-h^2 b(u,v)\right)  \toL \Ns\left(0 , \sigma^2(u,v) \right) , \label{eqn:normchat} \end{equation*}
where 
\begin{multline} b(u,v) =   \frac{1}{2}  \left\{\frac{\partial^2 c}{\partial u^2}(u,v)\phi^2(\Phi^{-1}(u))+ \frac{\partial^2 c}{\partial v^2}(u,v)\phi^2(\Phi^{-1}(v)) \right. \\ \ \ \ - 3\left( \frac{\partial c}{\partial u}(u,v) \Phi^{-1}(u)\phi(\Phi^{-1}(u))+ \frac{\partial c}{\partial v}(u,v) \Phi^{-1}(v)\phi(\Phi^{-1}(v))\right) \\
 + c(u,v)  \left(\{\Phi^{-1}(u)\}^2+\{\Phi^{-1}(v)\}^2-2 \right)  \bigg\}  \label{eqn:biaschat}
 \end{multline}
and $\displaystyle \sigma^2(u,v) = \frac{c(u,v)}{4\pi\phi(\Phi^{-1}(u))\phi(\Phi^{-1}(v))}$.
\end{theorem}
\begin{proof} See Appendix. \end{proof}

\ppn It is seen that, when $(u,v)$ approaches one of the boundaries, both the (asymptotic) bias and variance of the estimator tend to grow unboundedly. Indeed, $\sigma^2(u,v) \propto c(u,v)/(\phi(\Phi^{-1}(u))\phi(\Phi^{-1}(v)))$ and $b(u,v)$ includes the term $c(u,v)  \left(\{\Phi^{-1}(u)\}^2+\{\Phi^{-1}(v)\}^2-2 \right)$, and the functions $\Phi^{-1}(\cdot)$ and $1/\phi(\Phi^{-1}(\cdot))$ are unbounded. Thus, along the boundaries, $\hat{c}^{(\tau)}$ will work properly only over areas, if any, where $c$ approaches 0 very smoothly. Otherwise, $\hat{c}^{(\tau)}$ will typically show a very erratic behaviour (large variance) and will be prone to exploding (large positive bias), especially in the corners of $\Is$ in which $c$ is large. Figure \ref{fig:naivecopdens} illustrates these problems, from a typical sample of size $n=1000$ simulated from the Gaussian copula with correlation $\rho=0.3$ (left panel). The corresponding naive probit-transformation kernel estimator is shown in the middle 
panel. An unconstrained matrix $\HH_{ST} $ was used in (\ref{eqn:fSThat})/(\ref{eqn:naivefeas}) and chosen by the multivariate Normal Reference rule \citep{Chacon11}. Here, this is optimal: $C$ being a Gaussian copula, $f_{ST}$ is a bivariate normal density. Over the middle of the unit square, the estimator appears to work decently, but towards the boundaries the estimate shows coarse folds and, indeed, hypertrophies the peaks in the corners $(0,0)$ an $(1,1)$. Clearly, this estimator is not acceptable as-is. It is, therefore, not surprising that it has been reported not to perform well, see for instance the simulation study in \cite{Bouezmarni11}.

\begin{figure}
\centering
\includegraphics[width=\textwidth, trim=0 1.5cm 0 0]{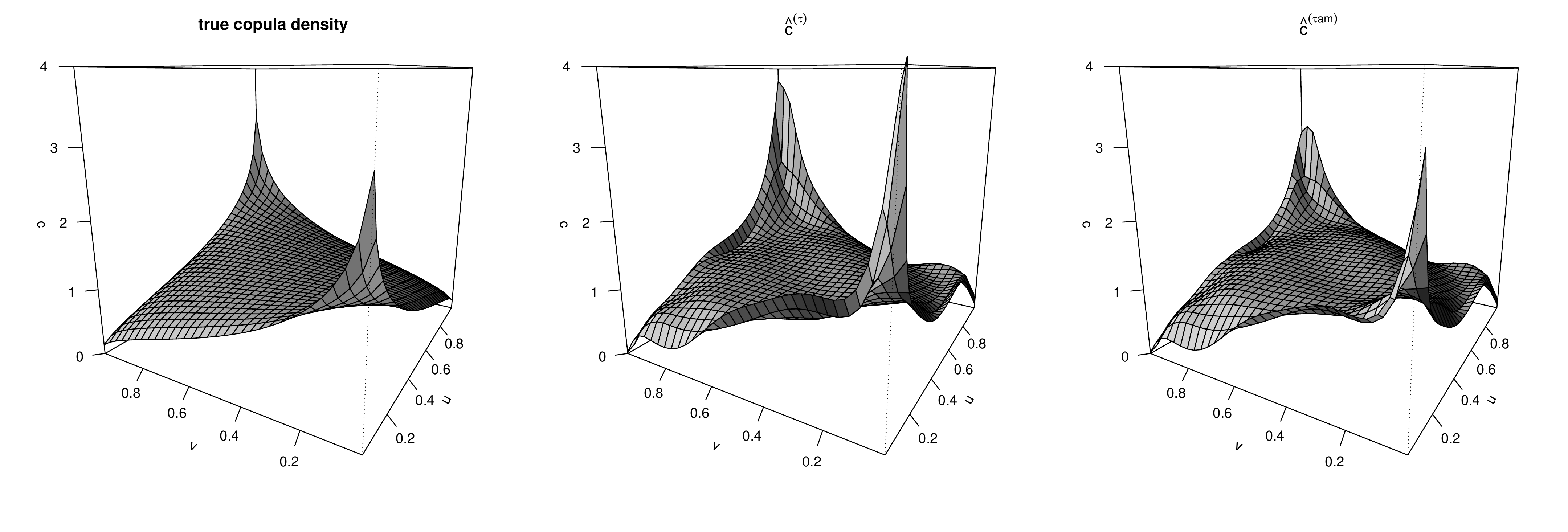}
\caption{True Gaussian copula density with $\rho=0.3$ (left), its naive probit-transformation kernel estimator from a typical random sample of size $n=1000$ (middle) and its amended naive probit-transformation kernel estimator from the same sample (right). The (unconstrained) bandwidth matrix $\HH_{ST}$ was chosen by the Normal Reference rule in the $(S,T)$-domain.}
\label{fig:naivecopdens}
\end{figure}

\ppn The third, unbounded term in (\ref{eqn:biaschat}) can, however, be easily adjusted for. Instead of (\ref{eqn:copratio}), take
\begin{equation} \hat{c}^{(\tau \text{am})}(u,v) = \frac{\hat{f}_{ST}(\Phi^{-1}(u),\Phi^{-1}(v))}{\phi(\Phi^{-1}(u))\phi(\Phi^{-1}(v))} \times \frac{1}{1+\frac{1}{2}h^2 \left(\{\Phi^{-1}(u)\}^2+\{\Phi^{-1}(v)\}^2-2 \right)}. \label{eqn:amendnaive} \end{equation}
For this `amended' version of $\hat{c}^{(\tau)}$, one can see that the asymptotic bias becomes proportional to
\begin{align*} b^{(\text{am})}(u,v)  = \frac{1}{2} &\left\{ \frac{\partial^2 c}{\partial u^2}(u,v)\phi^2(\Phi^{-1}(u))+ \frac{\partial^2 c}{\partial v^2}(u,v)\phi^2(\Phi^{-1}(v)) \right. \\ &  \left.  - 3\left( \frac{\partial c}{\partial u}(u,v) \Phi^{-1}(u)\phi(\Phi^{-1}(u))+ \frac{\partial c}{\partial v}(u,v) \Phi^{-1}(v)\phi(\Phi^{-1}(v))\right)\right\}. \end{align*}
In fact, the deterministic, multiplicative amendment in (\ref{eqn:amendnaive}) exactly makes it up for the third term in (\ref{eqn:biaschat}) in the asymptotic development, given that $(1+h^2)^{-1} = 1- h^2 +o(h^2)$ as $h \to 0$. The improvement is illustrated in Figure \ref{fig:naivecopdens} (right panel), where the amended version of the naive estimator computed on the same data set as in the middle panel is shown. The peaks in the corners $(0,0)$ and $(1,1)$ are now roughly of the right height. On the other hand, the wiggly appearance of the estimate along boundaries mostly remains, as the variance is not affected by the deterministic amendment. On a side note, the amendment implies that the estimator $\hat{c}^{(\tau \text{am})}(u,v)$ does not integrate to one over the unit square any more, which calls for a renormalisation such as $\hat{c}^{(\tau \text{am})}(u,v) \leftarrow \hat{c}^{(\tau \text{am})}(u,v) / \iint_{\Is} \hat{c}^{(\tau \text{am})}(u,v)\,du\,dv$. This is, however, 
frequent in other nonparametric density estimation procedures, and is not really a problem.

\subsection{Improved probit-transformation kernel copula density estimators} \label{sec:asymptpropimprov}

Now the asymptotic properties of the `improved' versions of the probit-transformation kernel copula density estimators are derived. Again, for convenience, the results are stated in the case where $\KK$ is a product of two univariate Gaussian kernels and $\HH_{ST} = h^2\II$, for some $h >0$, in (\ref{eqn:loclogpol}). The first version estimates the joint density $f_{ST}$ in the transformed domain by the local log-linear estimator $\tilde{f}_{ST}^{(1)}$. Consider first the `ideal' version $\tilde{f}_{ST}^{*(1)}$ of this estimator, computed on the genuine sample $\{(S_i,T_i),i=1,\ldots,n\}$. From \cite{Loader96}, one gets, for all $(s,t) \in \R^2$ at which $f_{ST}(s,t)$ is positive and admits continuous second-order partial derivatives, 
\begin{equation} \sqrt{nh^2}\left(\tilde{f}_{ST}^{*(1)}(s,t) - f_{ST}(s,t)-h^2 b^{(1)}_{ST}(s,t)\right)  \toL \Ns\left(0 , \left.\sigma^{(1)}_{ST}\right.^2(s,t) \right) , \label{eqn:normfSTtilde1} \end{equation}
where 
\[b^{(1)}_{ST}(s,t)  = \frac{1}{2}\left\{\left(\frac{\partial^2 f_{ST}}{\partial s^2}+\frac{\partial^2 f_{ST}}{\partial t^2}\right)(s,t) - \frac{1}{f_{ST}(s,t)}\,\left(\left\{\frac{\partial f_{ST}}{\partial s}\right\}^2+\left\{\frac{\partial f_{ST}}{\partial t}\right\}^2 \right)(s,t)\right\} \]
and $\left.\sigma^{(1)}_{ST}\right.^2(s,t) = \frac{f_{ST}(s,t)}{4\pi}$. Of course, if $f_{ST}(s,t)=0$ at some $(s,t)$, the singularity of the log-density cannot be accurately approximated by (\ref{eqn:locloglin}), but this is ruled out here by Assumption \ref{ass:copdens} which requires $c$ to be positive all over the interior of the unit square. By (\ref{eqn:fST}), this implies that $f_{ST}$ is positive over $\R^2$.

\ppn Define the `ideal' local log-linear probit-transformation kernel copula density estimator $\tilde{c}^{*(\tau, 1)}(u,v) = \tilde{f}_{ST}^{*(1)}(\Phi^{-1}(u),\Phi^{-1}(v))/\left(\phi(\Phi^{-1}(u)) \phi(\Phi^{-1}(v))\right)$. By the same token as for Theorem \ref{thm:chat}, in particular by using (\ref{eqn:fST}), (\ref{eqn:copratio}), (\ref{eqn:firstderfST}) and (\ref{eqn:secondderfST}) in (\ref{eqn:normfSTtilde1}), one can obtain
\begin{equation} \sqrt{nh^2}\left(\tilde{c}^{*(\tau, 1)}(u,v) - c(u,v)-h^2 b^{(1)}(u,v)\right)  \toL \Ns\left(0 , \left.\sigma^{(1)}\right.^2(u,v) \right) , \label{eqn:normctildestar1} \end{equation}
where 
\begin{multline} b^{(1)}(u,v) =  \frac{1}{2} \Bigg\{  \frac{\partial^2 c}{\partial u^2}(u,v)\phi^2(\Phi^{-1}(u))+ \frac{\partial^2 c}{\partial v^2}(u,v)\phi^2(\Phi^{-1}(v)) \\     
 -\frac{1}{c(u,v)}\left(\left\{\frac{\partial c}{\partial u}(u,v)\right\}^2\phi^2(\Phi^{-1}(u))+\left\{\frac{\partial c}{\partial v}(u,v)\right\}^2\phi^2(\Phi^{-1}(v)) \right) \\
 -\left(\frac{\partial c}{\partial u}(u,v) \Phi^{-1}(u) \phi(\Phi^{-1}(u)) + \frac{\partial c}{\partial v}(u,v) \Phi^{-1}(v) \phi(\Phi^{-1}(v)) \right)  -2c(u,v)  \Bigg\} \label{eqn:biasctilde1}  
\end{multline}
and $\displaystyle \left.\sigma^{(1)}\right.^2(u,v) = \frac{c(u,v)}{4\pi \phi(\Phi^{-1}(u))\phi(\Phi^{-1}(v))}$.

\ppn The next result ascertains that, like for the `naive' estimator, the asymptotic properties of $\tilde{c}^{(\tau,1)}$ are not affected by using pseudo-observations, and are consequently identical to those of the ideal version $\tilde{c}^{*(\tau,1)}$.

\begin{theorem} \label{thm:ctilde1} Under the assumptions of Proposition \ref{thm:fShat}, the `improved' local log-linear probit-transformation kernel copula density estimator $\tilde{c}^{(\tau,1)}$ at any $(u,v) \in (0,1)^2$ is such that
\begin{equation*} \sqrt{nh^2}\left(\tilde{c}^{(\tau,1)}(u,v) - c(u,v)-h^2 b^{(1)}(u,v)\right)  \toL \Ns\left(0 ,\left.\sigma^{(1)}\right.^2(u,v) \right) , \label{eqn:normctilde1} \end{equation*}
where $b^{(1)}(u,v)$ and $\left.\sigma^{(1)}\right.^2(u,v)$ are given above.  
\end{theorem}
\begin{proof} See Appendix. \end{proof}

Compared to the `naive' estimator, the variance is the same but the bias is significantly different. It is now automatically free from any unbounded terms. In fact, \cite{Hjort96} showed (their expression (7.3)) that the local log-linear and the standard kernel estimators in the $(S,T)$-domain ($\tilde{f}^{(1)}_{ST}$ and  $\hat{f}_{ST}$, respectively) satisfy
\begin{equation} \tilde{f}^{(1)}_{ST}(s,t) = \hat{f}_{ST}(s,t) \exp\left\{-\frac{1}{2} h^2 \left[\left(\frac{\partial \hat{f}_{ST}(s,t)/\partial s}{\hat{f}_{ST}(s,t)} \right)^2 + \left(\frac{\partial \hat{f}_{ST}(s,t)/\partial t}{\hat{f}_{ST}(s,t)} \right)^2 \right]\right\}. \label{eqn:corrslope} \end{equation}
This closed-form expression for $\tilde{f}^{(1)}_{ST}$ shows that it improves on the basic kernel estimator by adjusting for the local slopes. With (\ref{eqn:copratio}), (\ref{eqn:impctilde}) and an analogue of (\ref{eqn:firstderfST}) for hat versions, one can state a similar result in terms of the copula density estimators:
\begin{align*} \tilde{c}^{(\tau,1)}(u,v) = \hat{c}(u,v) \\ \times \exp\Bigg\{-\frac{1}{2} h^2 \Bigg[&\left(\frac{\partial \hat{c}(u,v)/\partial u}{\hat{c}(u,v)} \right)^2 \phi^2(\Phi^{-1}(u)) + \left(\frac{\partial \hat{c}(u,v)/\partial v}{\hat{c}(u,v)} \right)^2 \phi^2(\Phi^{-1}(v)) \\ & -2 \left\{\left(\frac{\partial \hat{c}(u,v)/\partial u}{\hat{c}(u,v)} \right) \Phi^{-1}(u) \phi(\Phi^{-1}(u))  +\left(\frac{\partial \hat{c}(u,v)/\partial v}{\hat{c}(u,v)} \right) \Phi^{-1}(v) \phi(\Phi^{-1}(v))\right\} \\ 
& +\left\{\Phi^{-1}(u) \right\}^2 +\left\{\Phi^{-1}(v) \right\}^2 \Bigg]\Bigg\}. \end{align*}
This reveals that, not only the local log-linear estimator $\tilde{c}^{(\tau,1)}$ attempts a correction for the slopes of $c$ like in (\ref{eqn:corrslope}), it actively acts on the boundary behaviour as well. Indeed, given that $\phi^2(\Phi^{-1}(\cdot))$ and $\Phi^{-1}(\cdot)\phi(\Phi^{-1}(\cdot))$ tend to 0 towards 0 and 1, the first four terms in the bracket in the previous expression will have little influence towards the boundaries (provided $c$ does not tend to 0 too sharply there). On the other hand, $\left\{\Phi^{-1}(u) \right\}^2 +\left\{\Phi^{-1}(v) \right\}^2$ tends to $+\infty$ very fast along boundaries (and all the more in the corners), hence $\hat{c}(u,v)$ is multiplied by something quickly tending to 0 there and this prevents it from exploding. This is, in fact, very similar to what the amendment in (\ref{eqn:amendnaive}) attempted, but is now achieved automatically. Figure \ref{fig:loclikcopdensfix} (middle panel) shows the estimate $\tilde{c}^{(\tau,1)}$ for 
the data set used in Figure \ref{fig:naivecopdens}. It used the cross-validation criterion discussed in Section \ref{sec:bandwidth} to select the matrix $\HH_{ST}$ in (\ref{eqn:loclogpol}). 

\begin{figure}
\centering
\includegraphics[width=\textwidth, trim=0 1.5cm 0 0]{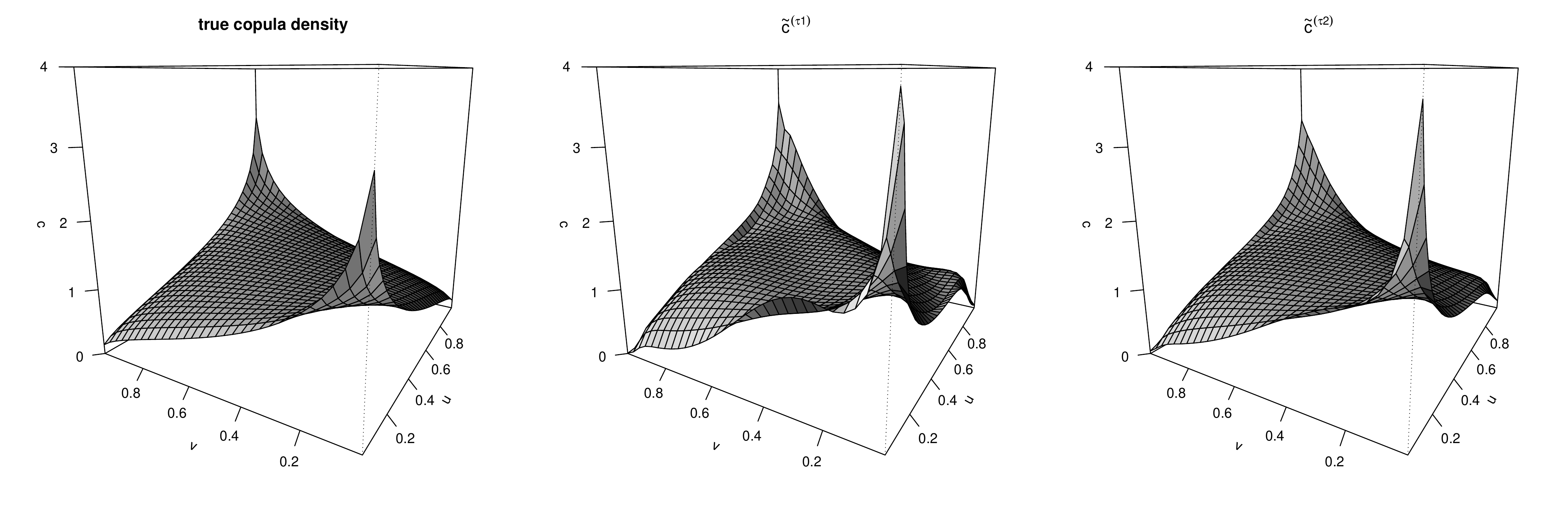}
\caption{True Gaussian copula density with $\rho=0.3$ (left), its local log-linear (middle) and log-quadratic (right) improved probit-transformation kernel estimators from the same sample ($n=500$) as in Figure \ref{fig:naivecopdens}. Both estimates use an unconstrained bandwidth matrix $\HH_{ST}$ chosen by cross-validation in the $(S,T)$-domain, see Section \ref{sec:bandwidth}.}
\label{fig:loclikcopdensfix}
\end{figure}

\ppn The second improved probit-transformation estimator is obtained when taking $p=2$ in (\ref{eqn:loclogpol}). Again, consider first the `ideal' estimator $\tilde{f}^{*(2)}_{ST}$, computed on the genuine sample $\{(S_i,T_i);i=1,\ldots,n\}$. Locally fitting a polynomial of a higher degree is known to reduce the asymptotic bias of the estimator, here from order $O(h^2)$ to order $O(h^4)$ \citep{Loader96,Hjort96}, sufficient smoothness of $f_{ST}$ permitting. Specifically, if $f_{ST}$ admits continuous fourth-order partial derivatives and is positive at $(s,t)$, then 
\begin{equation} \sqrt{nh^2}\left(\tilde{f}_{ST}^{*(2)}(s,t) - f_{ST}(s,t)-h^4 b^{(2)}_{ST}(s,t)\right)  \toL \Ns\left(0 , \left.\sigma^{(2)}_{ST}\right.^2(s,t) \right) , \label{eqn:normfSTtilde2} \end{equation}
where $\displaystyle \left.\sigma^{(2)}_{ST}\right.^2(s,t) = \frac{5}{2} \, \frac{f_{ST}(s,t)}{4\pi}$ and
\begin{multline*} b^{(2)}_{ST}(s,t)  = -\frac{1}{8}f_{ST}(s,t) \\ \times \left\{\left(\frac{\partial^4 g}{\partial s^4}+\frac{\partial^4 g}{\partial t^4}\right) +4\left( \frac{\partial^3 g}{\partial s^3} \frac{\partial g}{\partial s} + \frac{\partial^3 g}{\partial t^3} \frac{\partial g}{\partial t} + \frac{\partial^3 g}{\partial s^2 \partial t} \frac{\partial g}{\partial t} + \frac{\partial^3 g}{\partial s \partial t^2} \frac{\partial g}{\partial s}\right) + 2 \frac{\partial^4 g}{\partial s^2 \partial t^2} \right\}(s,t), \end{multline*}
with $g(s,t) = \log f_{ST}(s,t)$. Starting from $g(s,t) = \log c(\Phi(s),\Phi(t)) + \log \phi(s) + \log \phi(t)$, tedious algebraic differentiation provides all partial derivatives of $g$ up to order four in terms of $c$ and its partial derivatives up to order four. Naturally, $c$ will be assumed to admit continuous fourth-order partial derivatives.
\begin{assumption} \label{ass:cop4} The copula density $c(u,v) =(\partial^2 C/\partial u \partial v)(u,v)$ admits continuous fourth-order partial derivatives on the interior of the unit square $\Is$.
\end{assumption} 

As previously, it readily follows from (\ref{eqn:normfSTtilde2}) that
\begin{equation*} \sqrt{nh^2}\left(\tilde{c}^{*(\tau, 2)}(u,v) - c(u,v)-h^4 b^{(2)}(u,v)\right)  \toL \Ns\left(0 , \left.\sigma^{(2)}\right.^2(u,v) \right) , \label{eqn:normctildestar2} \end{equation*}
where $\displaystyle \left.\sigma^{(2)}\right.^2(u,v) = \frac{5}{2} \, \frac{c(u,v)}{4\pi \phi(\Phi^{-1}(u))\phi(\Phi^{-1}(v))}$  and $b^{(2)}(u,v)$ is an expression of the same type as (\ref{eqn:biasctilde1}), this time involving the partial derivatives of $c$ up to the fourth order. Like above, it can be shown that resorting to pseudo-observations does not affect these properties. This, however, requires a condition on the bandwidth ($h \sim n^{-a}, a \in (0,1/6)$) slightly stronger than previously. Given that the bias order is reduced to $O(h^4)$, the optimal bandwidth order is now seen to be $h \sim n^{-1/10}$, so that the bandwidth requirement does still include that optimal order.

\begin{theorem} \label{thm:ctilde2} Under the assumptions of Proposition \ref{thm:fShat} and Assumption \ref{ass:cop4}, if $h \sim n^{-a}$ with $a \in (0,1/6)$ as $n \to \infty$, the `improved' local log-quadratic probit-transformation kernel copula density estimator $\tilde{c}^{(\tau,2)}$ at any $(u,v) \in (0,1)^2$ is such that
\begin{equation*} \sqrt{nh^2}\left(\tilde{c}^{(\tau,2)}(u,v) - c(u,v)-h^4 b^{(2)}(u,v)\right)  \toL \Ns\left(0 ,\left.\sigma^{(2)}\right.^2(u,v) \right) , \label{eqn:normctilde2} \end{equation*}
where $b^{(2)}(u,v)$ and $\left.\sigma^{(2)}\right.^2(u,v)$ are described above.  
\end{theorem}
\begin{proof} See Appendix.
\end{proof}

For seek of conciseness, the expression of $b^{(2)}(u,v)$ is not given here (it is made up of several dozens of terms). The interesting point about it, though, is that, unlike (\ref{eqn:biasctilde1}) which shows a last term $-c(u,v)$, all terms of $b^{(2)}(u,v)$ are proportional to $\{\Phi^{-1}(u)\}^\alpha\{\Phi^{-1}(v)\}^\beta \phi^\gamma(\Phi^{-1}(u))\phi^\delta(\Phi^{-1}(v))$, for some non-negative integral powers $\alpha, \beta,\gamma,\delta$, and all those functions tend to 0 as $u,v \to 0/1$. Hence, $b^{(2)}(u,v)$ may actually tend to 0 towards the boundaries, and the bias there be actually of order $o(h^4)$. Again, this will be the case where $c$ does not tend to 0 too sharply when approaching the boundary. The expression of the variance is the same as that for $\hat{c}^{(\tau)}(u,v)$ and $\tilde{c}^{(\tau,1)}(u,v)$, except that it has been inflated by a factor $5/2$. This inflation factor is, also, a well-known  feature in local polynomial modelling when fitting a higher-degree polynomial \citep[
Section 3.3.1]{Fan96}.

\ppn Interestingly, {\it ad-hoc} techniques for reducing the bias of kernel estimators from $O(h^2)$ to $O(h^4)$, e.g.\ higher-order kernels or multiplicative adjustment, have long been an active research topic \citep{Jones97}. Yet, few of those methods have actually taken hold. The main reason is that the demonstrated improvement is an asymptotic result, which usually goes unnoticed for sample sizes one typically has in practice while implying interpretability issues (e.g.\ negative density estimates when using higher-order kernels) and computational burden. It is, therefore, worth stressing that here combining probit transformation and local log-quadratic density estimation in the $(S,T)$-domain achieves that bias reduction with no real extra complications compared to other estimators and fixes the boundary bias in an automatic way. Furthermore, these improvements are visible even in moderately large sample size, as the simulation study in Section \ref{sec:sim} will show. The most obvious and practically 
relevant effect of this bias reduction is that a larger bandwidth can be used without oversmoothing. This results in smoother estimates, visually more pleasant. This is clear in Figure \ref{fig:loclikcopdensfix} (right panel), where $\tilde{c}^{(\tau,2)}$ is shown for the same data set as previously. Again, the bandwidth matrix in (\ref{eqn:loclogpol}) was chosen via the cross-validation method suggested in Section \ref{sec:bandwidth}.

\subsection{Improved probit-transformation kernel copula density estimators with $k$-NN bandwidth} \label{subsec:kNN}

Theorems \ref{thm:ctilde1} and \ref{thm:ctilde2} reveal that the combination of probit transformation and local likelihood methods mostly cures the boundary bias problems for kernel copula density estimation. However, the fact remains that the suggested estimators have a variance behaving like
\begin{equation} \var(\tilde{c}^{(\tau,p)}(u,v)) = C_p \ \frac{c(u,v)}{4\pi nh^2 \phi(\Phi^{-1}(u))\phi(\Phi^{-1}(v))}  + o((nh^2)^{-1}), \label{eqn:varinf} \end{equation} 
where $C_1 = 1$ and $C_2 = 5/2$, as $n \to \infty$. Hence, $\var(\tilde{c}^{(\tau,p)}(u,v))$ tends to grow unboundedly when $(u,v)$ approaches one of the boundaries. Note that this is also the case for other copula density estimators attempting to correct the boundary bias, see for instance \cite[Chapter 4]{Blumentritt12} and \cite{Janssen13} who obtain similar unbounded boundary variance for the Beta kernel and the Bernstein estimators. 

\ppn Facing the same situation in the univariate case, \cite{Geenens13} suggested to use $k$-Nearest-Neighbor ($k$-NN) type bandwidth in the transformed domain. Although barely used for standard kernel density estimation, $k$-NN bandwidths appeared totally appropriate in \cite{Geenens13}'s framework and, indeed, managed to stabilise the variance towards the boundaries. This is also the case in this setting as can be understood heuristically as follows. 

\ppn  Again, assume that the smoothing matrix $\HH_{ST}$ in (\ref{eqn:loclogpol}) is diagonal, but instead of taking $\HH_{ST} = h^2\II$ for some fixed value $h$, take a local smoothing matrix defined as $\Hs^{(k)}_{ST}(s,t) = D^2_k(s,t) \II$, where $D_k(s,t)$ is the Euclidean distance between $(s,t)$ and the $k$th closest observation out of the sample (\ref{eqn:pseudosamp}) in $\R^2$. Now it is $k$, or equivalently $\alpha = k/n$, that will play the role of the smoothing parameter in lieu of $h$. If $\KK$ had a compact support, $\alpha$ would be the proportion of observations actively entering the estimation of $f_{ST}$ at any $(s,t)$ -- the interpretation roughly holds for the Gaussian kernel as well. 
Of course, $D_k(s,t)$ depends on the sample and is a random quantity. Along the same lines as in \cite{Mack79}, one can show that  $\E\left(1/D_k(s,t)\right) \simeq \frac{\pi f_{ST}(s,t)}{\alpha}$ and, together with $\var(\tilde{f}_{ST}^{(p)}(s,t)|D_{k}(s,t)) \simeq C_p \ \frac{f_{ST}(s,t)}{4\pi nD_{k}(s,t)}$, that
\begin{equation*} \var(\tilde{f}_{ST}^{(p)}(s,t)) \simeq C_p \ \frac{f^2_{ST}(s,t)}{4 n\alpha}.  \label{eqn:varknn2}\end{equation*}
Now, through (\ref{eqn:copratio}), one directly gets, for all $(u,v) \in (0,1)^2$,
\[\var(\tilde{c}^{(\tau,p)}(u,v)) \simeq C_p \ \frac{c^2(u,v)}{4 n\alpha}.   \]
Unlike (\ref{eqn:varinf}), this is no more proportional to $1/\{\phi(\Phi^{-1}(u))\phi(\Phi^{-1}(v))\}$ which grows unboundedly towards boundaries. This results in estimates more stable and much smoother towards the borders of $\Is$. 

\ppn In fact, given that the (long) tails of $\tilde{f}^{(p)}_{ST}$ in the transformed domain becomes the (short) boundary regions of $\tilde{c}^{(\tau,p)}$ in $\Is$ through the compressing back-transformation $(u=\Phi(s),v=\Phi(t))$, $\tilde{f}^{(p)}_{ST}$ must have {\it very smooth} tails in $\R^2$ to produce suitably smooth boundary behaviour for $\tilde{c}^{(\tau,p)}$. This is exactly what is achieved by using a $k$-NN bandwidth in the $(S,T)$-domain: local likelihood density estimators using $k$-NN bandwidth are, indeed, known to produce smoother estimates in the tails than their fixed-bandwidth counterparts, avoiding the occurrence of `spurious bumps'. Hence the appropriateness of the method here.

\ppn This is illustrated in Figure \ref{fig:loclikcopdensknn}. Again, the previous simulated data set was used to produce the two `improved' probit-transformation kernel copula density estimates shown in the middle (local log-linear) and the right panel (local log-quadratic), but this time a $k$-NN-type unconstrained bandwidth matrix $\Hs^{(k)}_{ST}(s,t)$ was used (see Section 4 for details). Compared to Figure \ref{fig:loclikcopdensfix}, the estimates are much smoother along the boundaries now. For instance, using a $k$-NN-bandwidth mostly corrects the kink previously observed in the $(1,0)$ corner.  It is particularly clear for $\tilde{c}^{(\tau,2)}$. The value of $\alpha$ selected for the case $p=1$ was 0.1871, that for the case $p=2$ was 0.4976. Again, the bias order reduction implied by local log-quadratic modelling allows a larger smoothing parameter to be used. As a result, this estimate $\tilde{c}^{(\tau,2)}$ (right panel) has a smooth and visually pleasant appearance, but without oversmoothing. In 
fact, it is barely distinguishable from the true copula density (left panel). It happens that the estimator $\tilde{c}^{(\tau,2)}$, when used in conjunction with a $k$-NN-type bandwidth, is strikingly good at recovering the shape of the underlying copula density while maintaining a visually pleasant amount of smoothness, see also Section \ref{sec:realdat}. 


\begin{figure}
\centering
\includegraphics[width=\textwidth, trim=0 1.5cm 0 0]{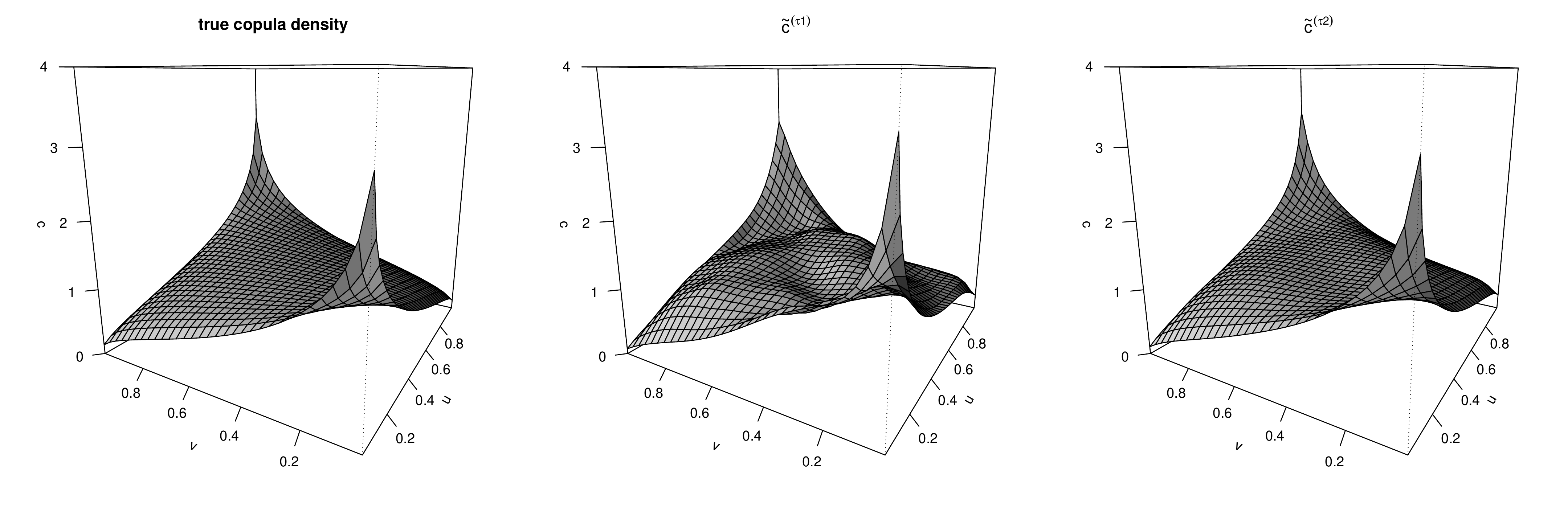}
\caption{True Gaussian copula density with $\rho=0.3$ (left), its local log-linear (middle) and log-quadratic (right) improved probit-transformation kernel estimators from the same sample ($n=500$) as in Figure \ref{fig:naivecopdens}. Both estimates use an unconstrained bandwidth matrix of type $k$-NN chosen by cross-validation in the $(S,T)$-domain, see Section \ref{sec:bandwidth}.}
\label{fig:loclikcopdensknn}
\end{figure}

\section{Bandwidth choice} \label{sec:bandwidth}

The behaviour of kernel estimators is known to be crucially dependent on their smoothing parameter, whose choice in practice is unanimously recognised as a very difficult problem, especially in more than one dimension. Here an effective way for selecting a suitable bandwidth matrix \[\HH_{ST} = \begin{pmatrix} h_1^2 & h_{12} \\ h_{12} & h^2_2 \end{pmatrix} \]
in (\ref{eqn:loclogpol}) is suggested. It can be understood that the diagonal elements $h^2_{1}$ and $h^2_{2}$ of $\HH_{ST}$ quantify the amount of smoothing applied in the directions of the main $s$- and $t$-axes, hence their values will drive the overall smoothness of the resulting estimate $\tilde{f}_{ST}^{(p)}$ and eventually that of $\tilde{c}^{(\tau,p)}$. On the other hand, $h_{12}$ sets the direction along which that smoothing mostly takes place in $\R^2$. For instance, if $\KK$ is the bivariate Gaussian kernel, the local weights around $(s,t) \in \R^2$ will be set by the elliptical contour lines of the $\Ns_2\left((s,t)^t,\HH_{ST}\right)$-distribution. If $f_{ST}$ is stretched along a particular direction of $\R^2$, which will be the case if $c$ itself is so on the unit square, it is greatly beneficial to the estimator that smoothing be applied in that direction \citep{Duong05}, and so $h_{12}$ should be selected accordingly. If this is not the case, in particular if $S$ and $T$ are uncorrelated, 
then $h_{12}$ may be set to 0. This motivates to separate the problem of selecting $h_1$ and $h_2$ from that of selecting $h_{12}$. The idea developed here looks for achieving this, in a way close in spirit to pre-sphering the observations \citep[Section 4.6]{Wand95}.

\ppn Consider the principal components decomposition of the $(n \times 2)$-`data matrix' $\Xi \doteq (\hat{S}_i,\hat{T}_i)_{i=1}^n$. By construction, the $\hat{S}_i$'s and $\hat{T}_i$'s are centred, hence $\hat{Q}_i$ and $\hat{R}_i$, the score of the $i$th observation on the first and second principal components, are given by
\begin{equation} \hat{Q}_i = W_{11}\hat{S}_i + W_{12} \hat{T}_i, \qquad \hat{R}_i = W_{21}\hat{S}_i + W_{22} \hat{T}_i, \label{eqn:PCA} \end{equation}
where $W_1=(W_{11},W_{12})^t$ and $W_2=(W_{21},W_{22})^t$ are the eigenvectors of $\Xi^T \Xi$. Given that the transformation 
\begin{equation} \binom{Q}{R}  = \begin{pmatrix} W_{11} & W_{12} \\ W_{21} & W_{22} \end{pmatrix} \binom{S}{T} \doteq  \WW \binom{S}{T} \label{eqn:linequiv} \end{equation}
is only a linear reparametrisation of $\R^2$, an estimate of $f_{ST}$ can be readily obtained from an estimate of the density of $(Q,R)$, say $f_{QR}$. In addition, it is well known that the samples $\{\hat{Q}_i\}$ and $\{\hat{R}_i\}$ are uncorrelated, hence estimating $f_{QR}$ via any kernel method from the sample $\{(\hat{Q}_i,\hat{R}_i);i=1,\ldots,n\}$ can be based on a diagonal bandwidth matrix $\HH_{QR} = \diag(h_Q^2,h_R^2)$ with little side effect. An idea is then to select $h_Q$ and $h_R$ independently via univariate procedures. Denote $\tilde{f}^{(p)}_Q$ and $\tilde{f}^{(p)}_R$ ($p=1,2$), the local log-polynomial estimators for the density of $Q$ and $R$, respectively, based on the samples $\{\hat{Q}_i\}$ and $\{\hat{R}_i\}$ (see equations (6) and (7) in \cite{Loader96}). Of course, $\tilde{f}^{(p)}_Q$ only depends on one bandwidth $h_Q$ and $\tilde{f}^{(p)}_R$ only depends on another bandwidth $h_R$. Then, $h_Q$ can be selected via cross-validation \citep[Section 5.3.3]{Loader99} as 
\begin{equation} h_Q = \arg \min_{h >0} \left\{\int_{-\infty}^\infty \left\{\tilde{f}^{(p)}_Q(q)\right\}^2\,dq - \frac{2}{n} \sum_{i=1}^n \tilde{f}^{(p)}_{Q(-i)}(\hat{Q}_i) \right\}, \label{eqn:hcv} \end{equation}
where, as usual in cross-validation procedures, $\tilde{f}^{(p)}_{Q(-i)}$ is the `leave-one-out' version of $\tilde{f}^{(p)}_{Q}$ computed on all the observations but $\hat{Q}_i$. The value of $h_R$ can be found similarly, and $h_Q$ and $h_R$ can be plugged into $\HH_{QR}$ for proceeding to bivariate estimation. However, optimal bandwidths for univariate density estimation are usually smaller than those for bivariate density estimation of $f_{QR}$. For the case $p=1$ (local log-linear estimator), the asymptotic optimal bandwidth order is $n^{-1/5}$ for univariate density estimation and $n^{-1/6}$ in two dimensions. For the case $p=2$ (local log-quadratic estimator), the asymptotic optimal bandwidth order is $n^{-1/9}$ for univariate density estimation and $n^{-1/10}$ in two dimensions. Hence, a fair choice for the bandwidth matrix for estimating $f_{QR}$ seems to be 
\[\HH_{QR} = K_n^{(p)}\begin{pmatrix} h_Q^2 & 0 \\ 0 & h^2_R \end{pmatrix}, \]
with $h_Q$ and $h_R$ the two bandwidths found above by (univariate) cross-validation, and $K^{(1)}_n = n^{1/15}$ in the local log-linear case and $K^{(2)}_n = n^{1/45}$ in the local log-quadratic case. The estimate of $f_{ST}$ can finally be obtained by linear back-transformation of the estimate of $f_{QR}$ from the $(Q,R)$-domain to the $(S,T)$ domain. It must be noted, though, that, again due to (\ref{eqn:linequiv}), this exactly amounts to directly estimating $f_{ST}$ from $\{(\hat{S}_i,\hat{T}_i);i=1,\ldots,n\}$ using the bandwidth matrix
\begin{equation*} \HH_{ST} = \WW^{-1} \begin{pmatrix} h_Q^2 & 0 \\ 0 & h^2_R \end{pmatrix}  \WW^{-1}. \label{eqn:CV} \end{equation*}

\ppn When the estimator is to be built on a $k$-NN-type bandwidth matrix, the procedure is very similar. The data are transformed into the sample $\{(\hat{Q}_i,\hat{R}_i);i=1,\ldots,n\}$ via (\ref{eqn:PCA}). Then, a suitable value of $\alpha$ in the $Q$-direction is computed as 
\begin{equation} \alpha_Q = \arg \min_{\alpha \in (0,1)} \left\{\int_{-\infty}^\infty \left\{\tilde{f}^{(p)}_Q(q)\right\}^2\,dq - \frac{2}{n} \sum_{i=1}^n \tilde{f}^{(p)}_{Q(-i)}(\hat{Q}_i) \right\}, \label{eqn:alphacv} \end{equation}
i.e.\ exactly as in (\ref{eqn:hcv}) but this time $\tilde{f}^{(p)}_Q$ is the estimator based on the $k$-NN bandwidth $\alpha$. The value $\alpha_R$ is computed in the same way. Denote $\kappa = \alpha_Q/\alpha_R$. Then, the squared norm of a vector in the $(Q,R)$-domain will be computed as 
\begin{equation} \|(q,r)\|^2  = q^2 + \kappa^2 r^2. \label{eqn:distkappa} \end{equation}
The factor $\kappa^2$ naturally adjusts, through the obtained close-to-optimal values of $\alpha_Q$ and $\alpha_R$, for the potential discrepancy in local geometry in the $q$- and $r$-directions. The bivariate estimation of $f_{QR}$ at any $(q,r) \in \R^2$ is carried out using the $k = K_n^{(p)} \times \alpha_Q \times n $ nearest neighbours of $(q,r)$, these being determined by the above distance. Here, $K_n^{(1)} = n^{-2/15}$ and $K_n^{(2)} = n^{-4/45}$, again for accounting for the difference in optimal $\alpha$-orders in one and two dimensions. Finally, the estimate of $f_{ST}$ is obtained by inverse linear transformation or, as set out in the fixed-bandwidth case, directly from the sample $\{(\hat{S}_i,\hat{T}_i);i=1,\ldots,n\}$ using an appropriate Mahalanobis-like distance. The main difference is that here, the employed distance makes use, through $\kappa$ in (\ref{eqn:distkappa}), of relevant information in terms of optimal smoothing, not only in terms of the covariance structure of $\{(\hat{S}_i,\hat{
T}_i);i=1,\ldots,n\}$ like the usual Mahalanobis distance. In this setting, the `smoothing parameters' vector is, therefore, $(\alpha_Q,\kappa)$.

\ppn It is acknowledged that this procedure may lack of sufficient theoretical support. For instance, it is known that pre-sphering the observations in the process of selecting the bandwidth matrix is justified only if the underlying density $f_{ST}$ is bivariate normal, that is, in this framework, if $c$ is the Gaussian copula. Likewise, choosing $h_Q$ and $h_R$ (or $\alpha_Q$ and $\alpha_R$) independently via univariate procedures would be suitable in theory only if $Q$ and $R$ were independent, not only uncorrelated. The need for a correcting factor $K_n^{(p)}$ in the above bandwidth expressions may also seem like nothing less than a heuristic, {\it ad-hoc} correction. Having said this, it was found that this way of doing gave very reliable results, as illustrated in Figures \ref{fig:loclikcopdensfix} and \ref{fig:loclikcopdensknn}. This will be even more obvious in the simulation study and the real data analysis detailed in the next sections. In addition, the suggested procedure, based on a twofold 
univariate cross-validation optimisation problem, is more stable numerically than one based on optimising a full, bivariate cross-validation criterion. This technique seems, therefore, an acceptable choice for selecting the bandwidth matrix in practice.

\section{Simulation study} \label{sec:sim}

Here, Monte Carlo simulations results are presented to compare the practical behaviour of the probit-transformation estimators with that of their main competitors. All computations have been carried out using the R software and its freely available packages. Specifically, 12 estimators were considered:
\begin{itemize}  \itemsep0em 
\item[$\cdot$] the `mirror reflection' estimator, denoted $\hat c^{(m)}$ below, as suggested in \cite{Gijbels90}. It was the first attempt at nonparametric copula density estimation, and remains a common choice for (ostensibly) correcting boundary bias. It will, therefore, be taken as benchmark. A first bandwidth matrix was obtained from the `augmented' data set (made up of $9n$ `observations' spread over an area 9 times bigger than $\Is$) via the Normal reference rule, then the final matrix was obtained by multiplying the former by $(1/9)^{2/3} \simeq 0.23$ for adjusting for the effective sample size and range;
\item[$\cdot$] the `naive' probit-transformation estimator $\hat c^{(\tau)}$ (\ref{eqn:naivefeas}) and its amended version $\hat c^{(\tau \text{am})}$, whose idea is exposed in Section \ref{subsec:naive} (for a general, non-diagonal matrix $\HH_{ST}$, the amendment takes a slightly more complicated form than (\ref{eqn:amendnaive})). The bandwidth matrix $\HH_{ST}$ was selected via a direct plug-in method \citep{Duong03} in the transformed domain $(S,T)$;
\item[$\cdot$] the improved probit-transformation estimators $\tilde c^{(\tau,1)}$ and $\tilde c^{(\tau,2)}$, given by (\ref{eqn:impctilde}), based on a $k$-NN-type bandwidth matrix selected via cross-validation as described at the end of Section \ref{sec:bandwidth}. As already observed in \cite{Geenens13} in the univariate case, when based on a fixed-bandwidth matrix these estimators performed a little less well, so the results are not shown here. The optimisation problems (\ref{eqn:loclogpol}) (local log-polynomial estimation of $f_{ST}$) and (\ref{eqn:alphacv}) ($k$-NN bandwidth selection) were solved using the relevant functions of the R package {\tt locfit}. A R package directly implementing these improved probit-transformation estimators is in preparation;
\item[$\cdot$] the Beta kernel estimator \citep{Charpentier07}, denoted $\hat c^{(\beta)}$, with \cite{Chen99}'s further bias correction. Two smoothing parameters were considered: $h=0.02$ ($\hat c_1^{(\beta)}$) and $h=0.05$ ($\hat c_2^{(\beta)}$);
\item[$\cdot$] the Bernstein copula density estimator \citep{Bouezmarni10,Bouezmarni11,Janssen13}, denoted $\hat c^{(B)}$. Two smoothing parameters were considered: $k=15$ ($\hat c_1^{(B)}$) and $k=30$ ($\hat c_2^{(B)}$);
\item[$\cdot$] the penalised hierarchical $B$-splines estimator \citep{Kauermann13}, denoted $\hat c^{(p)}$, computed by the function {\tt pencopula} in the eponymous R package. The vector of penalty coefficients was set to $\lambda=(10,10)$ ($\hat c_1^{(p)}$), $\lambda=(100,100)$ ($\hat c_2^{(p)}$), and $\lambda=(1000,1000)$ ($\hat c_3^{(p)}$). The parameters $d$ and $D$ were set to $4$ and $8$, according to \cite{Kauermann13}'s simulations study. 
\end{itemize}

\ppn $M=1,000$ independent random samples $\{(U_i,V_i);i=1,\ldots,n\}$ of sizes $n=200$, $n=500$ and $n=1000$ were generated from each of the following copulas:
\begin{itemize}  \itemsep0em 
\item[$\cdot$] the independence copula (i.e., $U_i$'s and $V_i$'s drawn independently);
\item[$\cdot$] the Gaussian copula, with parameters $\rho=0.31$, $\rho=0.59$ and $\rho=0.81$;
\item[$\cdot$] the Student $t$-copula with 10 degrees of freedom, with parameters $\rho=0.31$, $\rho=0.59$ and $\rho=0.81$;
\item[$\cdot$] the Student $t$-copula with 4 degrees of freedom, with parameters $\rho=0.31$, $\rho=0.59$ and $\rho=0.81$;
\item[$\cdot$] the Frank copula, with parameter $\theta=1.86$, $\theta=4.16$ and $\theta=7.93$;
\item[$\cdot$] the Gumbel copula, with parameter $\theta=1.25$, $\theta=1.67$ and $\theta=2.5$; 
\item[$\cdot$] the Clayton copula, with parameter $\theta=0.5$, $\theta=1.67$ and $\theta=2.5$.
\end{itemize}
For each family of copulas, the considered three values of the parameter roughly correspond to Kendall's $\tau$'s equal to 0.2, 0.4 and 0.6, respectively. Of course, all the estimations only made use of the pseudo-observations, i.e.\ the normalised ranks of the observations in the initially generated samples $\{U_i;i=1,\ldots,n\}$ and $\{V_i;i=1,\ldots,n\}$.

\ppn In order to assess the quality of the fit of an estimator $\hat{c}$ for a given copula density $c$, the Mean Integrated $L^2$-Error $\displaystyle{\E\left(\iint_\Is (\hat c(u,v)-c(u,v))^2 dudv\right)}$ was estimated by the average over the $M=1,000$ Monte Carlo replications of \[ISE(\hat{c}) \simeq \frac{1}{(N+1)^2} \sum_{k_1=1}^N \sum_{k_2=1}^N \left\{\hat{c}\left(\frac{k_1}{N+1},\frac{k_2}{N+1} \right)-c\left(\frac{k_1}{N+1},\frac{k_2}{N+1} \right) \right\}^2\]
with $N=64$. The approximated MISE can be found in Tables \ref{tab:L2:square-200}, \ref{tab:L2:square-500} and \ref{tab:L2:square-1000} for the three considered sample sizes. Note that, for ease of reading and interpretation, all the values are relative to the (approximated) MISE of the benchmark mirror estimator $\hat c^{(m)}$. For reference, the effective MISE of $\hat c^{(m)}$ is reported in italics in the last column of the table (which is, therefore, not on the same scale as the other values).

\begin{table}[ht]
\centering
\begin{tabular}{r|rr|rr|rr|rr|rrr||r}
  \hline
 \fbox{$n=200$} & $\hat c^{(\tau)}$ & $\hat c^{(\tau \text{am})}$  & $\tilde c^{(\tau,1)}$ & $\tilde c^{(\tau,2)}$ & $\hat c_1^{(\beta)}$ & $\hat c_2^{(\beta)}$ &  $\hat c_1^{(B)}$ &  $\hat c_2^{(B)}$ & $\hat c_{p}^{(1)}$ & $\hat c_{p}^{(2)}$ & $\hat c_{p}^{(3)}$ & $\hat{c}^{(m)}$\\ 
  \hline
Indep & 3.63 & 2.10 & 2.41 & 1.39 & 10.19 & 20.94 & 3.24 & 6.39 &  1.53 & 0.50 & {\bf 0.32} & {\it 0.02} \\ 
  Gauss2 & 2.52 & 1.46 & 1.55 & 0.92 & 6.81 & 13.55 & 2.18 & 4.14 &  1.01 & {\bf 0.53} & {\bf 0.49} &{\it 0.03} \\ 
  Gauss4 & 1.17 & 0.67 & 0.56 & {\bf 0.31} & 2.57 & 4.64 & 0.99 & 1.44 &  0.64 & 0.94 & 1.92 &{\it 0.08} \\ 
  Gauss6 & 0.50 & 0.30 & 0.16 & {\bf 0.08} & 0.88 & 1.15 & 0.69 & 0.53 &  0.76 & 1.25 & 2.16 & {\it 0.37} \\ 
  Std(10)2 & 2.05 & 1.15 & 1.28 & {\bf 0.72} & 5.30 & 10.95 & 1.71 & 3.18 &  0.96 & {\bf 0.77} & 0.90 & {\it 0.03} \\ 
  Std(10)4 & 0.76 & 0.53 & 0.42 & {\bf 0.24} & 1.90 & 3.22 & 0.87 & 1.08 &  0.71 & 1.04 & 1.77 & {\it 0.12} \\ 
  Std(10)6 & 0.33 & 0.28 & {\bf 0.13} & {\bf 0.10} & 0.78 & 0.85 & 0.68 & 0.48 &  0.82 & 1.22 & 1.89 & {\it 0.51} \\ 
  Std(4)2 & 1.12 & 0.76 & 0.81 & {\bf 0.57} & 2.83 & 5.74 & 1.12 & 1.74 &  0.87 & 1.00 & 1.27 & {\it 0.07} \\ 
  Std(4)4 & 0.44 & 0.41 & 0.32 & {\bf 0.25} & 1.23 & 1.79 & 0.76 & 0.71 &  0.79 & 1.12 & 1.59 & {\it 0.22} \\ 
  Std(4)6 & {\bf 0.19} & 0.28 & {\bf 0.15} & {\bf 0.16} & 0.74 & 0.60 & 0.73 & 0.51 &  0.88 & 1.17 & 1.60 & {\it 0.82} \\ 
  Frank2 & 3.54 & 2.00 & 2.17 & 1.28 & 9.06 & 18.22 & 2.81 & 5.53 &  1.20 & 0.37 & {\bf 0.30} & {\it 0.02} \\ 
  Frank4 & 2.74 & 1.41 & 1.28 & 0.88 & 5.40 & 10.41 & 1.83 & 3.14 &  {\bf 0.55} & 0.85 & 2.99 & {\it 0.03} \\ 
  Frank6 & 1.31 & 0.62 & {\bf 0.50} & {\bf 0.51} & 1.73 & 2.92 & 1.05 & 1.09 &  0.58 & 1.66 & 4.03 &{\it 0.13} \\ 
  Gumbel2 & 1.14 & 0.79 & 0.82 & {\bf 0.57} & 3.22 & 6.16 & 1.24 & 1.92 &  0.91 & 0.92 & 1.06 & {\it 0.06} \\ 
  Gumbel4 & 0.36 & 0.42 & {\bf 0.29} & {\bf 0.30} & 1.18 & 1.52 & 0.77 & 0.71 &  0.84 & 1.08 & 1.46 &{\it 0.26} \\ 
  Gumbel6 & {\bf 0.18} & 0.34 & {\bf 0.18} & 0.26 & 0.79 & 0.56 & 0.78 & 0.59 &  0.91 & 1.12 & 1.45 &{\it 1.08} \\ 
  Clayton2 & 1.13 & 0.76 & 0.72 & {\bf 0.50} & 3.16 & 6.54 & 1.16 & 1.87 &  0.89 & 0.96 & 1.25 &{\it 0.06} \\ 
  Clayton4 & {\bf 0.22} & 0.42 & {\bf 0.23} & 0.34 & 0.86 & 0.67 & 0.79 & 0.62 &  0.92 & 1.09 & 1.33 &{\it 0.75} \\ 
  Clayton6 & {\bf 0.19} & 0.43 & {\bf 0.22} & 0.32 & 0.82 & 0.51 & 0.82 & 0.66 &  0.95 & 1.08 & 1.26 & {\it 1.77} \\ 
   \hline
 \end{tabular}
\caption{(approximated) MISE relative to the MISE of the mirror-reflection estimator (last column), $n=200$. Bold values show the minimum MISE for the corresponding copula (non-significantly different values are highlighted as well). } 
\label{tab:L2:square-200}
\end{table}

\begin{table}[h]
\centering
\begin{tabular}{r|rr|rr|rr|rr|rrr||r}
  \hline
 \fbox{$n=500$} & $\hat c^{(\tau)}$ & $\hat c^{(\tau \text{am})}$  & $\tilde c^{(\tau,1)}$ & $\tilde c^{(\tau,2)}$ & $\hat c_1^{(\beta)}$ & $\hat c_2^{(\beta)}$ &  $\hat c_1^{(B)}$ &  $\hat c_2^{(B)}$ & $\hat c_{p}^{(1)}$ & $\hat c_{p}^{(2)}$ & $\hat c_{p}^{(3)}$ & $\hat{c}^{(m)}$\\ 
  \hline
Indep & 3.37 & 2.31 & 2.54 & 1.27 & 7.90 & 13.99 & 2.06 & 4.21 &  1.53 & 0.51 & {\bf 0.23} &{\it 0.01} \\ 
  Gauss2 & 2.22 & 1.47 & 1.63 & 0.78 & 5.24 & 8.82 & 1.41 & 2.59 &  0.98 & {\bf 0.63} & {\bf 0.68} &{\it 0.02} \\ 
  Gauss4 & 0.79 & 0.55 & 0.48 & {\bf 0.23} & 1.88 & 2.55 & 0.80 & 0.84 &  0.63 & 0.97 & 2.46 &{\it 0.06} \\ 
  Gauss6 & 0.31 & 0.24 & 0.13 & {\bf 0.06} & 0.74 & 0.60 & 0.70 & 0.39 &  0.73 & 1.23 & 2.54 &{\it 0.30 }\\ 
  Std(10)2 & 1.65 & 1.11 & 1.17 & {\bf 0.62} & 3.70 & 6.21 & 1.19 & 1.90 &  0.93 & 0.85 & 1.19 &{\it 0.02} \\ 
  Std(10)4 & 0.52 & 0.42 & 0.34 & {\bf 0.17} & 1.35 & 1.70 & 0.74 & 0.63 &  0.69 & 1.08 & 2.19 &{\it 0.10} \\ 
  Std(10)6 & 0.21 & 0.21 & {\bf 0.10} & {\bf 0.06} & 0.69 & 0.41 & 0.73 & 0.41 &  0.80 & 1.21 & 2.15 &{\it 0.44} \\ 
  Std(4)2 & 0.78 & 0.64 & 0.60 & {\bf 0.46} & 1.88 & 2.86 & 0.86 & 0.96 &  0.81 & 1.01 & 1.58 &{\it 0.05} \\ 
  Std(4)4 & 0.28 & 0.33 & 0.21 & {\bf 0.18} & 0.96 & 0.88 & 0.72 & 0.50 &  0.77 & 1.12 & 1.87 &{\it 0.18} \\ 
  Std(4)6 & {\bf 0.12} & 0.22 & {\bf 0.10} & {\bf 0.11} & 0.70 & 0.31 & 0.77 & 0.48 &  0.87 & 1.17 & 1.77 &{\it 0.74} \\ 
  Frank2 & 3.29 & 2.17 & 2.32 & 1.20 & 7.49 & 12.62 & 2.05 & 3.82 &  1.24 & {\bf 0.41} & {\bf 0.38} &{\it 0.01} \\ 
  Frank4 & 2.49 & 1.41 & 1.40 & 0.91 & 4.39 & 6.55 & 1.53 & 2.11 &  {\bf 0.58} & 0.80 & 4.50 &{\it 0.02} \\ 
  Frank6 & 1.02 & 0.54 & {\bf 0.43} & {\bf 0.43} & 1.44 & 1.71 & 1.13 & 0.81 &  0.49 & 1.62 & 5.67 &{\it 0.09} \\ 
  Gumbel2 & 0.83 & 0.71 & 0.65 & {\bf 0.47} & 2.16 & 3.20 & 0.90 & 1.09 &  0.87 & 0.98 & 1.30 &{\it 0.05} \\ 
  Gumbel4 & {\bf 0.25} & 0.35 & {\bf 0.21} & {\bf 0.23} & 0.94 & 0.72 & 0.76 & 0.53 &  0.82 & 1.09 & 1.64 &{\it 0.23} \\ 
  Gumbel6 & {\bf 0.11} & 0.26 & {\bf 0.12} & 0.18 & 0.77 & 0.36 & 0.82 & 0.57 &  0.92 & 1.12 & 1.56 &{\it 0.99} \\ 
  Clayton2 & 0.85 & 0.67 & 0.61 & {\bf 0.40} & 2.20 & 3.34 & 0.88 & 1.06 &  0.84 & 1.02 & 1.57 &{\it 0.04} \\ 
  Clayton4 & {\bf 0.15} & 0.32 & {\bf 0.14} & 0.21 & 0.79 & 0.37 & 0.79 & 0.56 &  0.91 & 1.09 & 1.43 &{\it 0.69} \\ 
  Clayton6 & {\bf 0.15} & 0.35 & {\bf 0.13} & 0.19 & 0.81 & 0.40 & 0.85 & 0.65 & 0.95 & 1.08 &  1.32 &{\it 1.67} \\ 
   \hline
 \end{tabular}
\caption{(approximated) MISE relative to the MISE of the mirror-reflection estimator (last column), $n=500$. Bold values show the minimum MISE for the corresponding copula (non-significantly different values are highlighted as well).} 
\label{tab:L2:square-500}
\end{table}

\begin{table}[h]
\centering
\begin{tabular}{r|rr|rr|rr|rr|rrr||r}
  \hline
 \fbox{$n=1000$} & $\hat c^{(\tau)}$ & $\hat c^{(\tau \text{am})}$  & $\tilde c^{(\tau,1)}$ & $\tilde c^{(\tau,2)}$ & $\hat c_1^{(\beta)}$ & $\hat c_2^{(\beta)}$ &  $\hat c_1^{(B)}$ &  $\hat c_2^{(B)}$ & $\hat c_{p}^{(1)}$ & $\hat c_{p}^{(2)}$ & $\hat c_{p}^{(3)}$ & $\hat{c}^{(m)}$\\ 
  \hline
Indep & 3.57 & 2.80 & 2.89 & 1.40 & 7.96 & 11.65 & 1.69 & 3.43 &  1.62 & 0.50 & {\bf 0.14} &{\it 0.01} \\ 
  Gauss2 & 2.03 & 1.52 & 1.60 & 0.76 & 4.63 & 6.06 & 1.10 & 1.82 &  0.98 & {\bf 0.66} & 0.89 &{\it 0.01} \\ 
  Gauss4 & 0.63 & 0.49 & 0.44 & {\bf 0.21} & 1.72 & 1.60 & 0.75 & 0.58 &  0.62 & 0.99 & 2.93 &{\it 0.05} \\ 
  Gauss6 & 0.21 & 0.20 & 0.11 & {\bf 0.05} & 0.74 & 0.33 & 0.77 & 0.37 &  0.72 & 1.21 & 2.83 &{\it 0.26} \\ 
  Std(10)2 & 1.36 & 1.06 & 1.04 & {\bf 0.55} & 3.07 & 3.98 & 0.96 & 1.24 &  0.86 & 0.87 & 1.48 &{\it 0.02} \\ 
  Std(10)4 & 0.41 & 0.37 & 0.28 & {\bf 0.15} & 1.22 & 1.00 & 0.74 & 0.46 &  0.68 & 1.08 & 2.51 &{\it 0.08} \\ 
  Std(10)6 & 0.15 & 0.18 & {\bf 0.08} & {\bf 0.05} & 0.71 & 0.24 & 0.79 & 0.41 &  0.84 & 1.21 & 2.36 &{\it 0.39} \\ 
  Std(4)2 & 0.61 & 0.56 & 0.50 & {\bf 0.40} & 1.57 & 1.80 & 0.78 & 0.67 &  0.75 & 1.01 & 1.88 &{\it 0.04} \\ 
  Std(4)4 & 0.21 & 0.27 & {\bf 0.17} & {\bf 0.15} & 0.88 & 0.51 & 0.75 & 0.42 &  0.75 & 1.12 & 2.07 &{\it 0.16} \\ 
  Std(4)6 & {\bf 0.09} & 0.17 & {\bf 0.08} & {\bf 0.09} & 0.70 & 0.19 & 0.82 & 0.47 &  0.90 & 1.17 & 1.90 &{\it 0.67} \\ 
  Frank2 & 3.31 & 2.42 & 2.57 & 1.35 & 7.16 & 9.63 & 1.70 & 2.95 &  1.31 & {\bf 0.45} & {\bf 0.49} &{\it 0.01} \\ 
  Frank4 & 2.35 & 1.45 & 1.51 & 0.99 & 4.42 & 4.89 & 1.49 & 1.65 &  {\bf 0.60} & 0.72 & 6.14 &{\it 0.01} \\ 
  Frank6 & 0.96 & 0.52 & {\bf 0.45} & {\bf 0.44} & 1.51 & 1.19 & 1.35 & 0.76 &  0.65 & 1.58 & 7.25 &{\it 0.07} \\ 
  Gumbel2 & 0.65 & 0.62 & 0.56 & {\bf 0.43} & 1.77 & 1.97 & 0.82 & 0.75 &  0.83 & 1.03 & 1.52 &{\it 0.04} \\ 
  Gumbel4 & {\bf 0.18} & 0.28 & {\bf 0.16} & {\bf 0.19} & 0.89 & 0.41 & 0.78 & 0.47 &  0.81 & 1.10 & 1.78 &{\it 0.21} \\ 
  Gumbel6 & {\bf 0.09} & 0.21 & {\bf 0.10} & 0.15 & 0.78 & 0.29 & 0.85 & 0.58 &  0.94 & 1.12 & 1.63 &{\it 0.93} \\ 
  Clayton2 & 0.63 & 0.60 & 0.51 & {\bf 0.34} & 1.78 & 1.99 & 0.78 & 0.70 &  0.79 & 1.04 & 1.79 &{\it 0.04} \\ 
  Clayton4 & {\bf 0.11} & 0.26 & {\bf 0.10} & {\bf 0.15} & 0.79 & 0.27 & 0.83 & 0.56 &  0.90 & 1.10 & 1.50 &{\it 0.65} \\ 
  Clayton6 & {\bf 0.11} & 0.28 & {\bf 0.08} & 0.15 & 0.82 & 0.35 & 0.88 & 0.67 &  0.96 & 1.09 & 1.36 &{\it 1.61} \\ 
   \hline
\end{tabular}
\caption{(approximated) MISE relative to the MISE of the mirror-reflection estimator (last column), $n=1000$. Bold values show the minimum MISE for the corresponding copula (non-significantly different values are highlighted as well). } 
\label{tab:L2:square-1000}
\end{table}

\ppn It turns out that the estimators $\tilde{c}^{(\tau,1)}$ and $\tilde{c}^{(\tau,2)}$ are clearly the best, overall, on this $L_2$-error criterion, and this for all sample sizes. They always dramatically improve on the Beta and Bernstein estimators, and they also do much better than the mirror reflection and the penalised B-splines estimators when the dependence is not close to null. When the depence is very low, $\hat{c}^{(m)}$ and $\hat{c}^{(p)}$ do better, which can be easily understood. It is well known that the mirror reflection estimator efficiently deals with boundary effects only when the partial derivatives of $c$ are 0 there (`shoulder'). It is, therefore, particularly appropriate for the independence copula ($c \equiv 1$) and other very flat copula densities such as Gaussian or Frank with low dependence. The penalised B-splines estimator does even better when using a huge penalty for roughness, for obvious reasons. In all other cases, and particularly when the copula density becomes unbounded in 
some corners (but not only), $\tilde{c}^{(\tau,1)}$ and $\tilde{c}^{(\tau,2)}$ dramatically outperform their competitors. In fact, mirror reflection and splines are not appropriate methods for estimating unbounded copula densities, and this is a real problem given that those are the most interesting cases in practice. By construction, the Beta kernel estimator always tends to be zero along boundaries (see for instance Figure \ref{fig:loss-alae-comp} below), hence its even worse performance in this framework. The Bernstein estimator does better than $\hat{c}^{(\beta)}$, but cannot really compete with $\tilde{c}^{(\tau,1)}$ and $\tilde{c}^{(\tau,2)}$. Of course, one can argue that the smoothing parameters used for $\hat{c}^{(\beta)}$, $\hat{c}^{(B)}$ and $\hat{c}^{(p)}$ have been selected mostly arbitrarily and are not adequate. This may be true, however, there is no simple, data-driven way of selecting those parameters, hence the choice was made subjectively exactly as a practitioner should have resolved to 
act. In addition, the 
above observations support that bad smoothing parameter choice is not the only reason for the poor performance of some of those estimators. Other evidence of that will be given in the next section on a real data set.

\ppn In general, the local log-quadratic estimator $\tilde{c}^{(\tau,2)}$ is doing better than the local log-linear $\tilde{c}^{(\tau,1)}$, which was expected from the theoretical results. A notable exception, though, is in presence of high tail dependence, i.e.\ when the copula density tends very quickly to $\infty$ at one of the corners of $\Is$, such as for Clayton and Gumbel copulas with high Kendall's $\tau$. In fact, the extra smoothness guaranteed by local log-quadratic estimation tends to prevent the estimator from growing too quickly in the corners, and this is thus slightly detrimental in those cases. The same comment holds true when comparing the naive estimator $\hat{c}^{(\tau)}$ to its amended version $\hat{c}^{(\tau \text{am})}$. Generally, $\hat{c}^{(\tau \text{am})}$ has lower MISE than $\hat{c}^{(\tau)}$, except in the above-mentioned cases of high tail dependence. Indeed, the amendment prevents the estimate from exploding, even when the true density does. In any case, these `naive' versions 
cannot really match the performance of the `improved' versions $\tilde{c}^{(\tau,1)}$ and $\tilde{c}^{(\tau,2)}$ on MISE, not to mention that their visual appearance is by far less pleasant. 

\ppn Finally, other criteria were considered for comparing the different estimators, such as $L_1$-error on the square $\iint_\Is |\hat c(u,v)-c(u,v)| dudv$, $L_1$- and $L_2$-error on the first diagonal ($u=v$) and on a side ($u=0.01$) of $\Is$, or $L_1$- and $L_2$-error at a given point in one of the corners of $\Is$ ($(u,v)=(0.01,0.01)$). These results are available on request from the authors. All show, to the same extent as above, the superiority of the improved probit-transformation estimators over their competitors.

\section{Real data analysis} \label{sec:realdat}

In this section the well-known `Loss-ALAE' dataset, reporting the indemnity payment ($X_i$'s) and allocated loss adjustment expense ($Y_i$'s) associated to $1,500$ losses from an insurance company, is considered. Analysed in \cite{Frees98,Klugman99} and \cite{Denuit06}, this dataset has since then become a classic in the copula literature. In particular, \cite{Frees98} mentioned that the Gumbel copula with $\hat{\theta}=1.453$ provides an excellent fit. The data set initially contains 34 censored observations, that were excluded here as the suggested estimators were not designed to take censorship into account. Using more advanced model selection ideas, \cite{Chen10} also found that the Gumbel copula (with the same parameter $\hat{\theta}$) fits the dataset (restricted to its complete cases) the best out of most of the usual parametric copula models. The aim here is to test the probit-transformation estimators $\tilde{c}^{(\tau,p)}$ ($p=1,2$) (and their competitors) against that parametric `gold standard', 
shown in Figure \ref{fig:loss-alae1} (up-left).

\begin{figure}[h]
\includegraphics[width=0.95\textwidth]{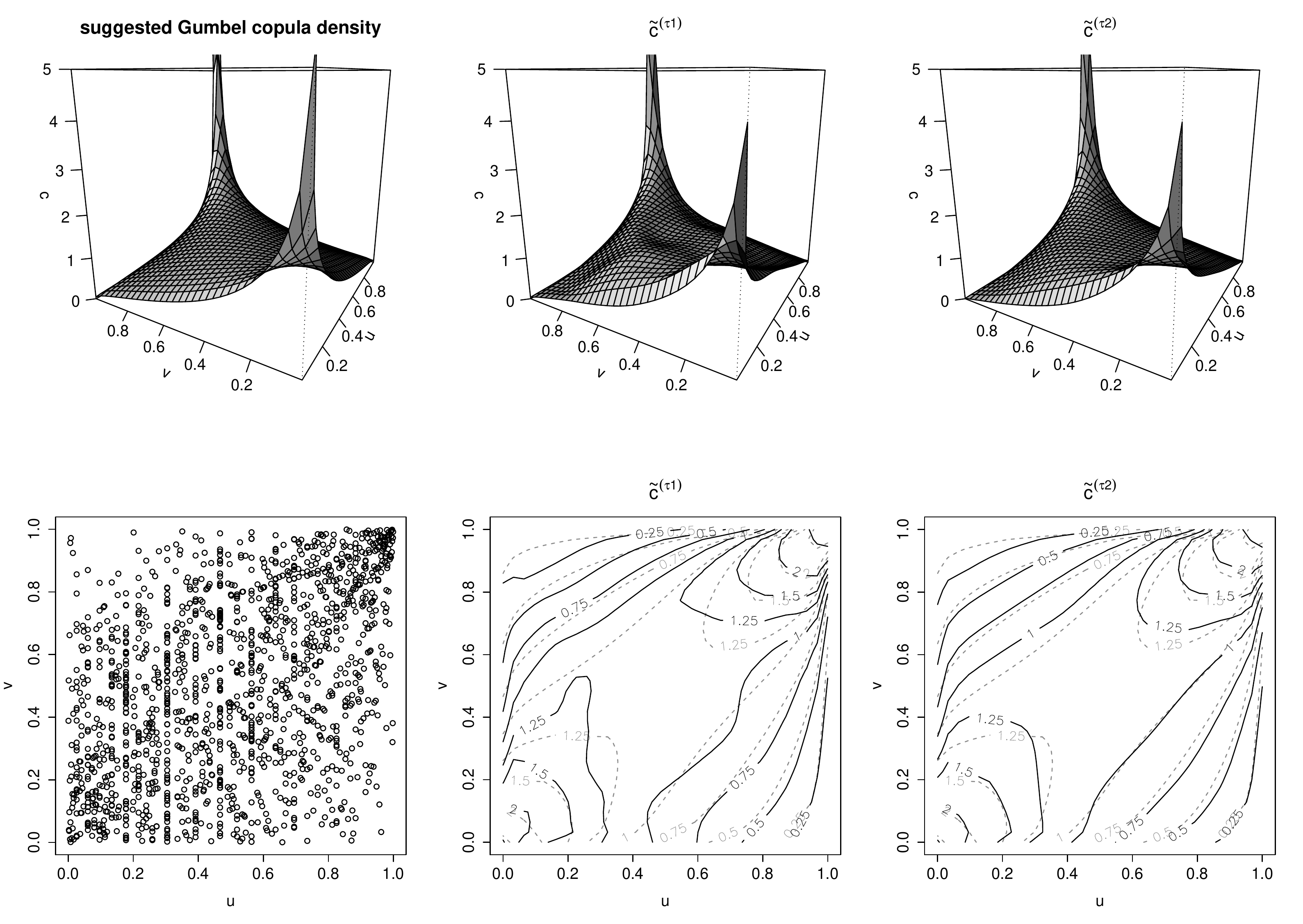}
\caption{Loss-ALAE dataset: suggested parametric copula density (Gumbel with parameter $\hat{\theta}=1.453$; upper-left panel) and probit-transformation estimates with $p=1$ (middle column) and $p=2$ (right column). The upper line shows 3-d views and the bottom line shows contour lines, superimposed on the Gumbel copula density contour lines. Pseudo-observations are shown in the bottom-left panel.}\label{fig:loss-alae1}
\end{figure}

\ppn The two probit-transformation estimators (local log-linear and local log-quadratic) were first fit to the data set. In both cases, a $k$-NN bandwidth matrix was used. Using the selection rule prescribed in Section \ref{sec:bandwidth}, the parameters $(\alpha, \kappa) = (0.24,1.28)$ for $p=1$ and $(\alpha, \kappa) = (0.51,1.01)$ for $p=2$ were obtained in an automatic manner. The estimator  $\tilde{c}^{(\tau,2)}$ is, again, very similar to the parametric fit. In particular, it has that very smooth and pleasant appearance of parametric estimates, while being based on a fully nonparametric procedure. Reproducing `parametric smoothness' without sacrificing any flexibility is, of course, a huge achievement for $\tilde{c}^{(\tau,2)}$. Naturally, $\tilde{c}^{(\tau,1)}$ is less smooth (smaller value of $\alpha$ than for $\tilde{c}^{(\tau,2)}$, for the reasons explained at the end of Section \ref{sec:asymptpropimprov}), but is still totally acceptable. Both nonparametric estimates suggest 
that the true underlying copula density decays towards the $(0,1)$-corner quicker than what the Gumbel model shows (this is particularly clear from the contour lines). Admittedly, there is no way of knowing what is the truth here. However, $\tilde{c}^{(\tau,1)}$ and $\tilde{c}^{(\tau,2)}$ are based only on the data (it is visually obvious that the upper-left corner of $\Is$ is much less endowed in data than the bottom-right corner), and not on any prior assumption. On the contrary, the Gumbel copula density is inherently symmetric in $u$ and $v$. The peak in the density at $(0,0)$ also appears less high on the nonparametric estimates than on the Gumbel copula density.

\ppn Figure \ref{fig:loss-alae-comp} shows the competitors on the same data set: the mirror reflection estimator, two Beta kernel estimators, two Bernstein estimators and two penalised $B$-splines estimators. Of course, $\hat{c}^{(m)}$ cannot cope with this unbounded copula density. For the other three methods, producing an estimate reasonably smooth required a value of the smoothing parameter ($h$ for Beta kernel estimators, $k$ for Bernstein estimators and $\lambda$ for penalised $B$-splines) preventing correct estimation of the peaks at $(0,0)$ and $(1,1)$. To get estimates showing a peak at $(1,1)$ of roughly the right magnitude, one needed to use smoothing parameters producing unacceptably undersmoothed estimation elsewhere on $\Is$, yet not even able to properly catch the peak at $(0,0)$. If the Gumbel copula density is assumed to be close to the truth for this data set, then there is no question that $\tilde{c}^{(\tau,1)}$ and $\tilde{c}^{(\tau,2)}$ are, by far, the best. This, also, illustrates that 
the results obtained in the simulations are not only due to bad smoothing parameter choices.

\begin{figure}[h]
\includegraphics[width=0.95\textwidth]{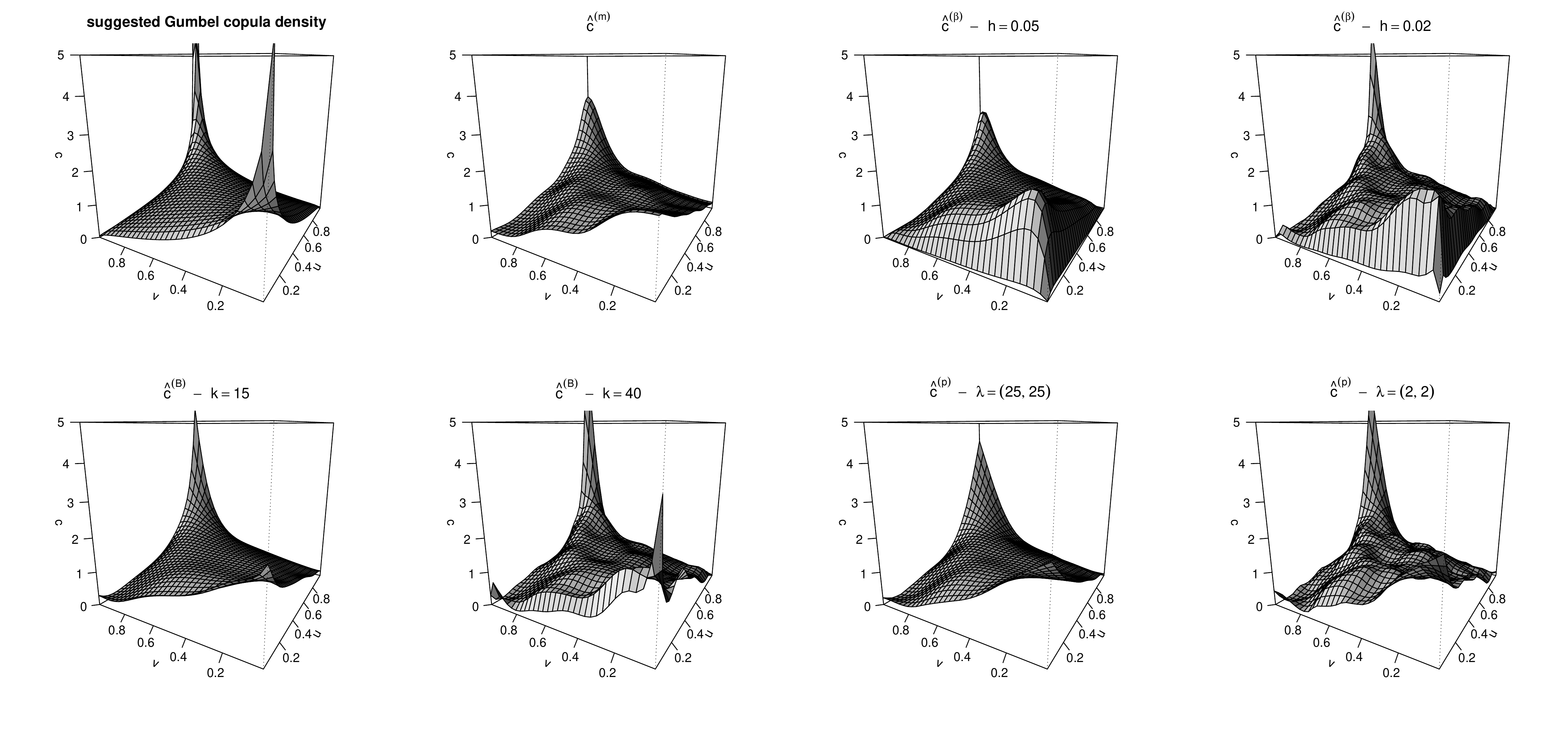}
\caption{Loss-ALAE dataset: suggested Gumbel copula density ($\hat{\theta} = 1.456$), mirror reflection estimator, Beta kernel estimators with $h=0.05$ and $h=0.02$, Bernstein estimators with $k=15$ and $k=40$ and penalized $B$-splines estimators with $d=4, D=8$ and with $\lambda = (25,25)$ and $\lambda = (2,2)$}\label{fig:loss-alae-comp}
\end{figure}

\section{Concluding remarks} \label{sec:ccl}

Development of efficient kernel-type methods for nonparametric copula modelling have been delayed owing mainly to the bounded support of copulas, namely the unit square $\Is$. It is, indeed, well known that kernel estimators heavily suffer from boundary bias issues, which are not trivial to fix. In this paper, a new kernel-type estimator for the copula density has been proposed. It is based on the probit-transformation idea suggested in \cite{Charpentier07} and studied in full in the univariate case in \cite{Geenens13}. This `improved probit-transformation estimator' deals with boundary bias in a very natural way. In addition, it has been seen to easily cope with potentially unbounded copula densities, which are the common and interesting cases in copula modelling. An easy-to-implement selection rule for the necessary smoothing parameters has also been proposed. This procedure has been seen to be very stable and to give very good results in practice. In particular, a version of the estimator ($\tilde{c}
^{(\tau,2)}$ with $k$-NN-type bandwidth matrix) is able to reproduce the smooth and pleasant appearance of parametric models, while keeping the flexibility of fully nonparametric estimation procedures. A comprehensive simulation study has emphasised the very good practical performance of that estimator compared to its main competitors. 

\ppn Several important points remain to be studied, though. First, as of now, the theoretical properties of the estimator have been derived under the assumption of i.i.d.\ sampling, making use of the strong approximation for the empirical copula process provided by Proposition 4.2 of \cite{Segers12}. To the best of these authors' knowledge, this result has not been proved in the case of weakly dependent observations. It would be particularly significant to investigate this in a near future, given the predominant place recently found by copula modelling in the setting of time series, notably in finance. Other directions for future research would look for using the new copula density estimator in a variety of problems, for instance copula goodness-of-fit tests \citep{Fermanian05,Scaillet07} or nonparametric conditional density estimation \citep{Faugeras09}. Finally, it must be said that, in theory, the idea presented in this paper is not bound to the bivariate case but extends in a straightforward way to higher dimensional copulas as well. Of course, in practice, this is wise only within the limits allowed by the {\it curse of dimensionality}.

\section*{Acknowledgments}

The first author was supported by a Faculty Research Grant from the Faculty of Science, University of New South Wales (Australia). The second author acknowledges additional funding provided by the Natural Sciences and Engineering Research Council of Canada. The third author was supported by an A.R.C.\ contract from the {\it Communaut\'e Fran\c{c}aise de Belgique} and by the IAP research network grant nr. P7/06 of the Belgian government (Belgian Science Policy).

\appendix 
\section{Appendix} \label{app:proofs}

First a technical lemma, which may be of interest of its own, is stated.

\begin{lemma} \label{lem:unicont} Under Assumptions \ref{ass:FXY}-\ref{ass:copdens}, the density $f_{ST}$ of the vector $(S,T)=(\Phi^{-1}(U),\Phi^{-1}(V))$ is uniformly bounded on $\R^2$, and so are its partial derivatives up to the second order.
\end{lemma}
\begin{proof} 
From Assumption \ref{ass:copdens}, one easily obtains that, for all $(u,v) \in (0,1)^2$,
\begin{align}
 c(u,v) & \leq K_{00} \min\left(\frac{1}{u(1-u)},\frac{1}{v(1-v)} \right)  \notag \\
& = K_{00} \, \min\left(\frac{1}{u(1-u)},\frac{1}{v(1-v)} \right)^\alpha \min\left(\frac{1}{u(1-u)},\frac{1}{v(1-v)} \right)^{1-\alpha} \qquad \forall \alpha \in (0,1) \notag \\
& \leq K_{00} \, \min\left(\frac{1}{u(1-u)},\frac{1}{v(1-v)} \right)^\alpha \max\left(\frac{1}{u(1-u)},\frac{1}{v(1-v)} \right)^{1-\alpha} \notag \\
& \leq K_{00} \, \min\left(\frac{1}{u(1-u)},\frac{1}{v(1-v)} \right)^\beta \max\left(\frac{1}{u(1-u)},\frac{1}{v(1-v)} \right)^{\beta} \qquad \text{with } \beta = \max(\alpha,1-\alpha) \notag \\
& = K_{00} \,\left( \frac{1}{uv(1-u)(1-v)}\right)^\beta, \qquad \text{ for some } \beta \in (1/2,1). \label{eqn:unb}
\end{align}
In particular, when approaching the $(0,0)$-corner, the above implies that
\[\lim_{(u,v) \to (0,0)} u^\beta v^{\beta} c(u,v) = k_{00}, \]
for some constant $0 \leq k_{00} < \infty$ (and similar towards the boundaries and the other corners of $\Is$). Now, applying Theorem 1 of \cite{Lawlor12}, one can show that this implies that, for $i,j = 0,1,2$ s.t.\ $i+j \leq 2$, there exist constants $k_{ij}<\infty$ such that
\[\lim_{(u,v) \to (0,0)} u^{\beta+i} v^{\beta+j} \left|\frac{\partial^{i+j} c(u,v)}{\partial u^i \partial v^j} \right| = k_{ij} \]
(and again, similar results hold at boundaries and at the other corners of $\Is$). Given that $c(u,v)$ is assumed to be twice continuously differentiable everywhere on the interior of $\Is$ (i.e., $c(u,v)$ and all its mixed partial derivatives of the first two orders can only go unbounded towards the boundaries of $\Is$), this allows one to write that, there exist $\beta \in (1/2,1)$ and bounded constants $K_{ij}$ such that
\begin{equation} \left|\frac{\partial^{i+j} c(u,v)}{\partial u^i \partial v^j} \right| \leq  \ \frac{K_{ij}}{u^{\beta+i}(1-u)^{\beta+i}v^{\beta+j}(1-v)^{\beta+j}}, \label{eqn:derivunb} \end{equation}
for all $(u,v) \in (0,1)^2$. Then, (\ref{eqn:unb}) in (\ref{eqn:fST}) yields 
\[f_{ST}(s,t) \leq \frac{K_{00}\phi(s)\phi(t)}{\Phi(s)^{\beta}(1-\Phi(s))^{\beta}\Phi(t)^{\beta}(1-\Phi(t))^{\beta}} \]
for all $(s,t) \in \R^2$. Given that, from the known properties of the normal distribution, $\sup_{s \in \R} \frac{\phi(s)}{\Phi(s)^{\beta}(1-\Phi(s))^{\beta}}$ is bounded for any constant $\beta < 1$, there exists a constant $M$ such that $\sup_{(s,t) \in \R^2} f_{ST}(s,t) \leq  M < \infty$. Now from (\ref{eqn:firstderfST}) one can write
\[\left|\frac{\partial f_{ST}}{\partial s}(s,t)\right|  \leq \left|\frac{\partial c}{\partial u}(\Phi(s),\Phi(t))\right|\phi^2(s) \phi(t) + |s|c(\Phi(s),\Phi(t))  \phi(s)\phi(t). \]
From (\ref{eqn:unb}) and (\ref{eqn:derivunb}) with $i=1$, $j=0$, one gets
\[ \left|\frac{\partial f_{ST}}{\partial s}(s,t)\right|  \leq \frac{K_{10}\phi^2(s) \phi(t)}{\Phi(s)^{\beta+1}(1-\Phi(s))^{\beta+1}\Phi(t)^{\beta}(1-\Phi(t))^{\beta}}  + |s|\frac{K_{00}\phi(s)\phi(t)}{\Phi(s)^{\beta}(1-\Phi(s))^{\beta}\Phi(t)^{\beta}(1-\Phi(t))^{\beta}}. \]
Take $\gamma = (\beta + 1)/2 < 1$ and see that, as above, $\sup_{s \in \R}\frac{\phi(s)}{\Phi(s)^{\gamma}(1-\Phi(s))^{\gamma}}$ is bounded, so that  the first term is uniformly bounded on $\R^2$. The second term is uniformly bounded as well, as $\sup_{s \in \R}\frac{s\phi(s)}{\Phi(s)^{\beta}(1-\Phi(s))^{\beta}}$ is also bounded for any $\beta < 1$. All second-order partial derivatives of $f_{ST}$ can be uniformly bounded in the exact same way using (\ref{eqn:derivunb}) in (\ref{eqn:secondderfST}) and similar. \end{proof}

\subsection*{Proof of Proposition \ref{thm:fShat}}

Denote
\[C_n(u,v) = \frac{1}{n} \sum_{i=1}^n \indic{U_i \leq u,V_i \leq v}, \]
i.e.\ the `ideal' version of the empirical copula (\ref{eqn:empcop}) using genuine observations $(U_i,V_i)$'s, and the corresponding empirical process $\{\B_n(u,v):(u,v) \in \Is\}$, with
\[\B_n(u,v) = \sqrt{n}(C_n(u,v) - C(u,v)). \]
Also, define the process $\{\G_n(u,v):(u,v) \in \Is\}$ with
\begin{equation}
\G_n(u,v) = \B_n(u,v) - \frac{\partial C}{\partial u}(u,v) \B_n(u,1) - \frac{\partial C}{\partial v}(u,v) \B_n(1,v). \label{eqn:Gn} 
\end{equation}
\cite{Segers12} shows that, under Assumptions \ref{ass:FXY}-\ref{ass:copdens}, the empirical copula process $\C_n(u,v)$ (see (\ref{eqn:limitempcop})) and $\G_n(u,v)$ are such that
\begin{equation} \sup_{(u,v) \in \Is} |\C_n(u,v) - \G_n(u,v)| = O_{\text{a.s.}}\left(n^{-1/4}(\log n)^{1/2} (\log \log n)^{1/4}\right) \qquad  n \to \infty.\label{eqn:strong} \end{equation}

\ppn Now, see that
\[ \sqrt{nh^2} \left( \hat{f}_{ST}(s,t) - \E\left(\hat{f}^*_{ST}(s,t) \right)\right) = \frac{1}{h} \iint_\Is \phi\left(\frac{s-\Phi^{-1}(u)}{h}\right)\phi\left(\frac{t-\Phi^{-1}(v)}{h}\right)\,d\C_n(u,v).\]
Integrating by parts, as in the proof of Theorem 6 of \cite{Fermanian04}, one gets
\begin{align} \sqrt{nh^2} \left( \hat{f}_{ST}(s,t) - \E\left(\hat{f}^*_{ST}(s,t) \right)\right) &  \notag \\ = \frac{1}{h} \iint_\Is & \C_n(u,v) \,\phi'\left(\frac{s-\Phi^{-1}(u)}{h}\right)\phi'\left(\frac{t-\Phi^{-1}(v)}{h}\right)\,\frac{du}{h\phi(\Phi^{-1}(u))}\,\frac{dv}{h\phi(\Phi^{-1}(v))}  \notag\\
 = \frac{1}{h} \iint_\Is & \G_n(u,v) \,\phi'\left(\frac{s-\Phi^{-1}(u)}{h}\right)\phi'\left(\frac{t-\Phi^{-1}(v)}{h}\right)\,\frac{du}{h\phi(\Phi^{-1}(u))}\,\frac{dv}{h\phi(\Phi^{-1}(v))} \notag \\
& + R_n(s,t) \label{eqn:Gn2}
\end{align}
where 
\begin{align*} |R_n(s,t)| & \leq \frac{1}{h} \sup_{(u,v) \in \Is} | \C_n(u,v)- \G_n(u,v)| \left\{\int_0^1  \left|\phi'\left(\frac{s-\Phi^{-1}(u)}{h}\right)\right|\,\frac{du}{h\phi(\Phi^{-1}(u))}\right\}^2 \\
& =  \frac{1}{h} \sup_{(u,v) \in \Is} | \C_n(u,v)- \G_n(u,v)| \left\{\int_\R |z| \phi(z) \,dz \right\}^2 \\
& = O_{\text{a.s.}}\left(n^{-1/4}h^{-1}(\log n)^{1/2} (\log \log n)^{1/4}\right) \\ 
& = o_{\text{a.s.}}(1),\end{align*}
given that $\E(|\Ns(0,1)|)<\infty$ and the conditions on the bandwidth $h$. Call 
\[ J_{st,h}(u,v) = \phi\left(\frac{s-\Phi^{-1}(u)}{h}\right)\phi\left(\frac{t-\Phi^{-1}(v)}{h}\right),\]
so that
\[ dJ_{st,h}(u,v) = \phi'\left(\frac{s-\Phi^{-1}(u)}{h}\right)\phi'\left(\frac{t-\Phi^{-1}(v)}{h}\right) \frac{du}{h\phi(\Phi^{-1}(u))}\,\frac{dv}{h\phi(\Phi^{-1}(v))}.\]
Plugging (\ref{eqn:Gn}) in (\ref{eqn:Gn2}) yields
\begin{align*}
 \sqrt{nh^2} \left( \hat{f}_{ST}(s,t) - \E\left(\hat{f}^*_{ST}(s,t) \right)\right) = &\  \frac{1}{h} \iint_\Is \B_n(u,v) \,dJ_{st,h}(u,v)\\
 & - \frac{1}{h} \iint_\Is  \frac{\partial C}{\partial u}(u,v) \B_n(u,1) \,dJ_{st,h}(u,v) \\
 & - \frac{1}{h} \iint_\Is  \frac{\partial C}{\partial v}(u,v) \B_n(1,v)\,dJ_{st,h}(u,v) + R_{n}(s,t)\\ 
\doteq &\  A_n(s,t) + B_{n,1}(s,t) + B_{n,2}(s,t)+R_{n}(s,t).
\end{align*}
The process $\B_n(u,v)$ being the bivariate empirical process based on genuine observations, $A_n(s,t)$ is, in fact, $\sqrt{nh^2} \left( \hat{f}^*_{ST}(s,t) - \E\left(\hat{f}^*_{ST}(s,t) \right)\right)$, for which classical kernel smoothing theory states that 
\[A_n(s,t)  \toL \Ns\left( 0 , \frac{f_{ST}(s,t)}{4\pi} \right).   \]
The terms $B_{n,1}$ and $B_{n,2}$ can be worked out explicitly. This is done below for $B_{n,1}$ only ($B_{n,2}$ can be treated in the exact same way). Write 
\[B_{n,1}(s,t) = - \frac{1}{h} \int_0^1   \B_n(u,1) \,\phi'\left(\frac{s-\Phi^{-1}(u)}{h}\right) \left\{ \int_0^1 \frac{\partial C}{\partial u}(u,v)  \phi'\left(\frac{t-\Phi^{-1}(v)}{h}\right)\,\frac{dv}{h\phi(\Phi^{-1}(v))}\right\} \frac{du}{h\phi(\Phi^{-1}(u))} \] 
and proceed with 
\begin{align*} \psi(u) & \doteq \int_0^1 \frac{\partial C}{\partial u}(u,v)  \phi'\left(\frac{t-\Phi^{-1}(v)}{h}\right)\,\frac{dv}{h\phi(\Phi^{-1}(v))} \\ 
& = \int_\R \frac{\partial C}{\partial u}(u,\Phi(t-hz))  \phi'\left(z\right)\,dz\end{align*}
with the change of variable $z=h^{-1}(t-\Phi^{-1}(v))$. As $C(u,v) = F_{ST}(\Phi^{-1}(u),\Phi^{-1}(v))$, this is also
\begin{equation} \psi(u) = \frac{1}{\phi(\Phi^{-1}(u))} \int_\R \frac{\partial F_{ST}}{\partial s}(\Phi^{-1}(u),t-hz)  \phi'\left(z\right)\,dz. \label{eqn:psiu} \end{equation}
Taylor-expanding, one gets
\begin{multline*} \frac{\partial F_{ST}}{\partial s}(\Phi^{-1}(u),t-hz) = \frac{\partial F_{ST}}{\partial s}(\Phi^{-1}(u),t) - hz \frac{\partial^2 F_{ST}}{\partial s \partial t}(\Phi^{-1}(u),t) \\ + \frac{1}{2}h^2z^2 \frac{\partial^3 F_{ST}}{\partial s \partial t^2}(\Phi^{-1}(u),t) -  \frac{1}{6}h^3z^3 \frac{\partial^4 F_{ST}}{\partial s \partial t^3}(\Phi^{-1}(u),t-h^*z),  \end{multline*}
where $h^* \in (0,h)$, of which only odd powers of $z$ will remain in (\ref{eqn:psiu}) as $\phi'$ is an odd function. Hence,
\begin{align*} \psi(u) & = -\frac{1}{\phi(\Phi^{-1}(u))} \left(h f_{ST}(\Phi^{-1}(u),t) \int_\R z \phi'(z)\,dz +\frac{1}{6}h^3 \int_\R z^3 \frac{\partial^2 f_{ST}}{\partial t^2}(\Phi^{-1}(u),t-h^*z)\phi'(z)\,dz \right) \\
& =  \frac{1}{\phi(\Phi^{-1}(u))} \left(h f_{ST}(\Phi^{-1}(u),t) +\frac{1}{6}h^3 \psi^*(u) \right), \end{align*}
since $\int_\R z \phi'(z)\,dz = -1$, denoting $\psi^*(u) = -\int_\R z^3 \frac{\partial^2 f_{ST}}{\partial t^2}(\Phi^{-1}(u),t-h^*z)\phi'(z)\,dz $.

\ppn It follows that 
\begin{align*} B_{n,1}(s,t) =   - &\  \frac{1}{h} \int_0^1   \B_n(u,1) \,\phi'\left(\frac{s-\Phi^{-1}(u)}{h}\right)\psi(u) \frac{du}{h\phi(\Phi^{-1}(u))} \\
=   & - \  \int_0^1   \B_n(u,1) f_{ST}(\Phi^{-1}(u),t)\phi'\left(\frac{s-\Phi^{-1}(u)}{h}\right)\,  \frac{du}{h\phi^2(\Phi^{-1}(u))} \\
& - \frac{1}{6}h^2 \int_0^1   \B_n(u,1) \,\phi'\left(\frac{s-\Phi^{-1}(u)}{h}\right) \psi^*(u) \frac{du}{h\phi^2(\Phi^{-1}(u))} \\
\doteq \quad &  B_{n,11}(s,t) +B_{n,12}(s,t).
 \end{align*}
The change of variable $w = \frac{s-\Phi^{-1}(u)}{h}$ yields
\begin{align*} B_{n,11}(s,t) = &\  \int_\R \B_n(\Phi(s-hw),1) \phi'(w) \frac{f_{ST}(s-hw,t)}{\phi(s-hw)}\,dw \\
= & \ \B_n(\Phi(s),1) \frac{f_{ST}(s,t)}{\phi(s)} \int_\R \phi'(w) \,dw \\
&  - \frac{f_{ST}(s,t)}{\phi(s)} \int_\R (\B_n(\Phi(s),1)-\B_n(\Phi(s)-h\phi(s-h^*w),1))\phi'(w) \,dw \\
& + h \int_\R \B_n(\Phi(s-hw),1) \phi'(w) \frac{\partial}{\partial s}\left\{\frac{f_{ST}(s-h^{**}w,t)}{\phi(s-h^{**}w)}\right\}\,dw, \end{align*}
where $h^*,h^{**}$ are both between $0$ and $h$. The first term in this decomposition is 0, as $\phi'$ is an odd function. The second term tends to 0 as $n \to 0$. This is because $\B_n(\cdot,1)$ is nothing else than a usual univariate uniform empirical process. \citet{Einmahl87} studied its modulus of continuity and their Theorem 3.1(b) allows one to write, under these assumptions,
\[  \sup_{u,u' \in [0,1],|u-u'|\leq h} |\B_n(u,1)-\B_n(u',1)| = \sqrt{2h \log h^{-1}} \qquad \text{ almost surely as } n \to \infty.  \]
Hence the second term can readily be seen to be $O_{\text{a.s.}}(\sqrt{h \log h^{-1}})$ as $n \to \infty$, i.e.\ $o_{\text{a.s.}}(1)$ as $h \to 0$. The integral in the third term is $O_P(1)$. Given that $\frac{\partial}{\partial s}\left\{\frac{f_{ST}(s-h^{**}w,t)}{\phi(s-h^{**}w)}\right\} = \frac{\partial f_{ST}/\partial s}{\phi}(s-h^{**}w,t)+ (s-h^{**}w) \frac{f_{ST}}{\phi}(s-h^{**}w,t)$, one can write
\begin{multline*} \int_\R \phi'(w) \frac{\partial}{\partial s}\left\{\frac{f_{ST}(s-h^{**}w,t)}{\phi(s-h^{**}w)}\right\}\,dw \\ = - \int_\R w\frac{\phi(w)}{\phi(s-h^{**}w)} \frac{\partial f_{ST}}{\partial s}(s-h^{**}w,t)\,dw - \int_\R (s-h^{**}w)w\frac{\phi(w)}{\phi(s-h^{**}w)} f_{ST}(s-h^{**}w,t)\,dw. \end{multline*}
This is bounded for $n$ large enough, because $f_{ST}$ and $\frac{\partial f_{ST}}{\partial s}$ are uniformly bounded on $\R^2$ by Lemma \ref{lem:unicont}, and 
\[\int_\R \left|w^j \frac{\phi(w)}{\phi(s-h^{**}w)}\right|\,dw < \infty, \qquad j=1,2, \]
for $h^{**} < 1$, which eventually occurs as $n \to \infty$. As $\sup_{w} \B_n(\Phi(s-hw),1) = O_P(1)$, one can see that
\[|B_{n,11}(s,t)| = O_P(h) + O_{\text{a.s.}}(\sqrt{h \log h^{-1}}) = o_P(1) \qquad \text{ as } n \to \infty. \]
The same can be written for $|B_{n,12}(s,t)|$, given that $\psi^*(u)$ is a uniformly bounded function of $u$ on $[0,1]$, from Lemma \ref{lem:unicont}. Hence,
\[ \sqrt{nh^2} \left( \hat{f}_{ST}(s,t) - \E\left(\hat{f}^*_{ST}(s,t) \right)\right) \toL \Ns\left(0,\frac{f_{ST}(s,t)}{4\pi} \right). \]
The final result follows by seeing that, using classical kernel density estimation results,
\[ \E\left(\hat{f}^*_{ST}(s,t)\right) = f_{ST}(s,t)+\frac{1}{2}h^2\left(\frac{\partial^2 f_{ST}}{\partial s^2}(s,t)+\frac{\partial^2 f_{ST}}{\partial t^2}(s,t) \right) +o(h^2) \qquad \text{ as } n \to \infty.\]
\qed

\subsection*{Proof of Theorem \ref{thm:chat}}

Result (\ref{eqn:normfSThat}) can be written at $(s,t) = (\Phi^{-1}(u),\Phi^{-1}(v))$:
\begin{multline*} \sqrt{nh^2}\left(\hat{f}_{ST}(\Phi^{-1}(u),\Phi^{-1}(v)) - f_{ST}(\Phi^{-1}(u),\Phi^{-1}(v))-h^2 b_{ST}(\Phi^{-1}(u),\Phi^{-1}(v))\right)  \\ \toL \Ns\left(0 , \sigma_{ST}^2(\Phi^{-1}(u),\Phi^{-1}(v)) \right),\end{multline*}
which implies 
\begin{multline} \sqrt{nh^2}\left(\frac{\hat{f}_{ST}(\Phi^{-1}(u),\Phi^{-1}(v))}{\phi(\Phi^{-1}(u))\phi(\Phi^{-1}(v))} - \frac{f_{ST}(\Phi^{-1}(u),\Phi^{-1}(v))}{\phi(\Phi^{-1}(u))\phi(\Phi^{-1}(v))}-h^2 \frac{b_{ST}(\Phi^{-1}(u),\Phi^{-1}(v))}{\phi(\Phi^{-1}(u))\phi(\Phi^{-1}(v))}\right)  \\ \toL \Ns\left(0 , \frac{\sigma_{ST}^2(\Phi^{-1}(u),\Phi^{-1}(v))}{\phi^2(\Phi^{-1}(u))\phi^2(\Phi^{-1}(v))}\right). \label{eqn:toL} \end{multline}
Now, 
\[ \frac{\sigma_{ST}^2(\Phi^{-1}(u),\Phi^{-1}(v))}{\phi^2(\Phi^{-1}(u))\phi^2(\Phi^{-1}(v))} = \frac{f_{ST}(\Phi^{-1}(u),\Phi^{-1}(v))}{4\pi\phi^2(\Phi^{-1}(u))\phi^2(\Phi^{-1}(v))} = \frac{c(u,v)}{4\pi\phi(\Phi^{-1}(u))\phi(\Phi^{-1}(v))},\]
using (\ref{eqn:truecopratio}). In addition, from (\ref{eqn:secondderfST}), it is easily seen that $\frac{b_{ST}(\Phi^{-1}(u),\Phi^{-1}(v))}{\phi(\Phi^{-1}(u))\phi(\Phi^{-1}(v))}$ has the form (\ref{eqn:biaschat}). Using (\ref{eqn:truecopratio})/(\ref{eqn:copratio}) in (\ref{eqn:toL}) then concludes the proof. \qed

\subsection*{Proof of Theorem \ref{thm:ctilde1}}

From Proposition \ref{thm:fShat} one can conclude that 
\begin{equation} |\hat{f}_{ST}(s,t) - \hat{f}^*_{ST}(s,t)| = o_P((nh^2)^{-1/2}). \label{eqn:o} \end{equation}
Now, from (\ref{eqn:corrslope}), one has
\[ |\tilde{f}^{(1)}_{ST}(s,t) - \tilde{f}^{*(1)}_{ST}(s,t)| = \left|\hat{f}_{ST}(s,t) \exp\left\{-\frac{1}{2} h^2 \hat{\psi}_{ST}(s,t) \right\} - \hat{f}^*_{ST}(s,t)\exp\left\{-\frac{1}{2} h^2 \hat{\psi}^*_{ST}(s,t) \right\}\right|,
\]
where 
\[\hat{\psi}_{ST}(s,t) = \left(\frac{\partial \hat{f}_{ST}(s,t)/\partial s}{\hat{f}_{ST}(s,t)} \right)^2 + \left(\frac{\partial \hat{f}_{ST}(s,t)/\partial t}{\hat{f}_{ST}(s,t)} \right)^2 \]
and equivalently for the star version. Then,
\begin{align*} 
 |\tilde{f}^{(1)}_{ST}(s,t) - \tilde{f}^{*(1)}_{ST}(s,t)| \leq &\  |\hat{f}_{ST}(s,t) - \hat{f}^*_{ST}(s,t)|\exp\left\{-\frac{1}{2} h^2 \hat{\psi}_{ST}(s,t) \right\} \\
&  + \hat{f}^*_{ST}(s,t) \left|\exp\left\{-\frac{1}{2} h^2 \hat{\psi}_{ST}(s,t) \right\}-\exp\left\{-\frac{1}{2} h^2 \hat{\psi}^*_{ST}(s,t) \right\} \right|.
\end{align*}
The first term is $o_P((nh^2)^{-1/2})$, given (\ref{eqn:o}) and $\exp\left\{-\frac{1}{2} h^2 \hat{\psi}_{ST}(s,t) \right\} \overset{P}{\to} 1$ as $n \to \infty$ (i.e.\ as $h^2 \to 0$). 

\ppn The second term can be written 
\begin{equation} \frac{1}{2}\,h^2\,\hat{f}^*_{ST}(s,t) \exp\left\{-\frac{1}{2} h^2 \check{\psi}^*_{ST}(s,t) \right\} \left|\hat{\psi}_{ST}(s,t) -\hat{\psi}^*_{ST}(s,t) \right|, \label{eqn:2nd}\end{equation}
where $\check{\psi}^*_{ST}(s,t) $ is between $\hat{\psi}_{ST}(s,t)$ and $\hat{\psi}^*_{ST}(s,t)$. Now, it can be seen that
\begin{equation} \frac{\left|\hat{\psi}_{ST}(s,t) -\hat{\psi}^*_{ST}(s,t) \right|}{\left|\frac{\partial \hat{f}_{ST}}{\partial s}-\frac{\partial \hat{f}^*_{ST}}{\partial s} \right|(s,t) + \left|\frac{\partial \hat{f}_{ST}}{\partial t}-\frac{\partial \hat{f}^*_{ST}}{\partial t} \right|(s,t)} = O_P\left(1\right).\label{eqn:psi} \end{equation}  
From standard kernel arguments, it is known that, if $nh^4 \to \infty$ as $n \to \infty$, $\partial \hat{f}^*_{ST}(s,t)/\partial s$ is a consistent estimator of $\partial f_{ST}(s,t)/\partial s$, with variance $O((nh^4)^{-1})$ and bias $O(h^2)$. In a way very similar to the proof of Proposition \ref{thm:chat}, one can show that resorting to pseudo-observations does not change the statistical properties of that estimator of $\partial f_{ST}(s,t)/\partial s$. Hence $|\partial \hat{f}^*_{ST}(s,t)/\partial s-\partial \hat{f}_{ST}(s,t)/\partial s| = o_P((nh^4)^{-1/2})$ (and same for the partial derivatives with respect to $t$). It follows that $\left|\hat{\psi}_{ST}(s,t) -\hat{\psi}^*_{ST}(s,t) \right| = o_P\left((nh^4)^{-1/2}\right)$. With the factor $h^2$ in (\ref{eqn:2nd}), this means that the second term is of order $o_P(n^{-1/2})$, which is obviously $o_P((nh^2)^{-1/2})$. Then,
\[|\tilde{f}^{( 1)}_{ST}(s,t) - \tilde{f}^{*(1)}_{ST}(s,t)| = o_P((nh^2)^{-1/2}), \]
and since this order carries over to $|\tilde{c}^{(\tau,1)}(u,v) - \tilde{c}^{*(\tau, 1)}(u,v)|$ in a straightforward way, (\ref{eqn:normctildestar1}) also holds with $\tilde{c}^{(\tau,1)}(u,v)$ instead of $\tilde{c}^{*(\tau, 1)}(u,v)$. \qed

\subsection*{Proof of Theorem \ref{thm:ctilde2}}

It is very similar to the proof of Theorem \ref{thm:ctilde1}, based on the analogue of (\ref{eqn:corrslope}) for $\tilde{f}^{(2)}_{ST}(s,t)$ (i.e.\ the bivariate version of equation (5.2) in \cite{Hjort96}). It is, therefore, omitted. The reason why a condition on $h$ stronger than previously is needed is that here, the analogue of (\ref{eqn:psi}) involves the second order partial derivatives of $\hat{f}_{ST}$. To have those consistent for the corresponding partial derivatives of $f_{ST}$, indeed, requires $nh^6 \to \infty$. \qed


\end{document}